\pdfoutput=1
\documentclass[12pt,reqno]{amsart}
\usepackage{amsmath,amsfonts,amssymb,amscd,amsthm,amsbsy}
\usepackage{tensor}
\usepackage{comment}

\usepackage{courier}
\usepackage{wrapfig}
\usepackage{fancybox}

\usepackage[pdftex]{color}

\usepackage{multimedia}
\usepackage{graphicx,epstopdf}

\usepackage{appendix}
\usepackage{natbib}

\usepackage{cancel,soul,ulem}

\defcitealias{mL89}{Lane's}
\defcitealias{gG74}{Giacaglia's}

\numberwithin{equation}{section}

\newtheorem{theorem}{Theorem}
\newtheorem{meta-thm}[theorem]{Meta-Theorem}
\newtheorem{lemma}[theorem]{Lemma}

\newtheorem{proposition}[theorem]{Proposition}

\newtheorem{remark}[theorem]{Remark}

\newcommand\beq[1]{ \begin{equation}\label{#1} }
\newcommand{\eeq}{ \end{equation} }

\newcommand\beqa[1]{ \begin{eqnarray} \label{#1}}

\newcommand{\bfm}[1]{\mbox{\boldmath{$#1$}}}
\newcommand{\ds}{\displaystyle}

\newcommand{\eeqa}{ \end{eqnarray} }
\newcommand{\beqano}{ \begin{eqnarray*} }
\newcommand{\eeqano}{ \end{eqnarray*} }
\newcommand\equ[1]{{\rm (\ref{#1})}}

\def\F{{\mathcal F}}

\def\F{{\mathcal F}}

\def\G{{\mathcal G}}

\def\H{{\mathcal H}}

\def\R{{\mathcal R}}

\def\integer{{\mathbb Z}}

\begin{document}

\title[Lunisolar disturbing function and critical inclination secular resonance]
{Analytical development of the lunisolar disturbing function and the critical inclination secular resonance}

\author[A. Celletti]{Alessandra Celletti}
\address{
Department of Mathematics, University of Roma Tor Vergata, Via della Ricerca Scientifica 1,
00133 Roma (Italy)}
\email{celletti@mat.uniroma2.it}

\author[C. Gale\c s]{C\u at\u alin Gale\c s}
\address{
Department of Mathematics, Al. I. Cuza University, Bd. Carol I 11,
700506 Iasi (Romania)}
\email{cgales@uaic.ro}

\author[G. Pucacco]{Giuseppe Pucacco}
\address{
Department of Physics, University of Roma Tor Vergata, Via della Ricerca Scientifica 1,
00133 Roma (Italy)}
\email{pucacco@roma2.infn.it}

\author[A. J. Rosengren]{Aaron J. Rosengren}
\address{
IFAC-CNR, Via Madonna del Piano 10, 50019 Sesto Fiorentino (FI) (Italy)}
\email{a.rosengren@ifac.cnr.it}

\date{Received: date / Accepted: date}

\begin{abstract}
We provide a detailed derivation of the analytical expansion of
the lunar and solar disturbing functions. Although there exist
several papers on this topic, many derivations contain mistakes in
the final expansion or rather (just) in the proof, thereby
necessitating a recasting and correction of the original
derivation. In this work, we provide a self-consistent and
definite form of the lunisolar expansion. We start with Kaula's
expansion of the disturbing function in terms of the equatorial
elements of both the perturbed and perturbing bodies. Then we give
a detailed proof of Lane's expansion, in which the elements of the
Moon are referred to the ecliptic plane. Using this approach the
inclination of the Moon becomes nearly constant, while the
argument of perihelion, the longitude of the ascending node, and
the mean anomaly vary linearly with time.

We make a comparison between the different expansions and we
profit from such discussion to point out some mistakes in the
existing literature, which might compromise the correctness of the
results. As an application, we analyze the long--term motion of
the highly elliptical and critically--inclined Molniya orbits
subject to quadrupolar gravitational interactions. The analytical
expansions presented herein are very powerful with respect to
dynamical studies based on Cartesian equations, because they
quickly allow for a more holistic and intuitively understandable
picture of the dynamics.
\end{abstract}


\maketitle

\section{Introduction}

The most interesting and long--standing problem of celestial mechanics is that of determining the perturbing effects due to the gravitational force of bodies other than the central mass; more formally known as third--body perturbations. The Moon, its motion around the Earth disturbed by the gravitational attraction of the exceedingly large and relatively near Sun, presented one of the most complex problems within the Solar System. The third--body perturbation problem has also occupied a prominent place in modern celestial mechanics, from the study of the dynamical evolution of stellar and planetary systems to the orbital motions of small bodies and artificial satellites. The analytic methodology for computing and describing these perturbations almost invariably employs a theoretical development of the disturbing function---the negative potential function of the disturbing acceleration \citep{kEcM00,rM13}. The disturbing function plays a fundamental role in celestial mechanics, giving rise to the notion of separation of perturbing effects into periodic and secular variations and the distinction between fast and slow time variables \citep{cMsD99}.

The mathematical development of the lunar and solar effects on the motion of artificial Earth satellites was originally made by \citet{yK59} and \citet{pM59}, who expanded the disturbing function into a series of Legendre polynomials in the ratio of the radial distances (a small quantity for close satellite orbits). A more general and convenient series development was made by \citet{wK62}, whereby the Legendre polynomial is expanded using the addition theorem for spherical harmonics in terms of tesseral harmonics involving polar coordinates referred to the Earth's axis, and further expanded in terms of orbital elements relative to the celestial equator. This harmonic analysis of the perturbations involves a complicated coupling of the motion of the third body (Moon and Sun) with that of the satellite \citep{gC62}. \citet{yK66} noted that the results of these calculations are expressed most simply, in their dependence on the orbital elements of the Moon, when the latter are defined with respect to the plane of the ecliptic. Such a choice of reference planes permits us to consider the lunar inclination as a constant and the lunar argument of perigee and longitude of the ascending node as linear functions of time, provided that we are interested in determining perturbations of the first order only.

The mixed--reference--frame formalism requires a rotation of the
harmonic functions depending on the lunar position \citep{bJ65},
which, while being straightforward and formal in nature, is quite
laborious and error--prone. The first such attempt was made by
\citet{gG74} (see also \cite{gG80}). \citet{sH80}, following this
line of thought, presents a similar expansion, but replaces the
Kaula inclination functions with those of \citet{rA65, rA73} and
\citet{iZ64}. As first noted by \citet{mL89}, however, several
algebraic errors appear in the reckoning work of \citet{gG74}, so
that the form of the lunar series expansion presented therein is
incorrect. Accordingly, \citet{mL89} gives the complete
development showing the tedious and cumbersome nature of the
calculations, but omitting some of the more evident mathematical
details. Despite such a careful and detailed treatment, we
recently discovered some errors in the proof given by Lane - which
we correct here - although the final lunar series expansion given
in \citet{mL89} is indeed correct.

Of course, other series expansion formulations exist, the validity of which have not been questioned, such as those of \citet{rE74} and \citet{dC77}, based instead on the Hill--Brown lunar theory. Nevertheless, the use of  \citetalias{gG74} formalism offers in many respects decided advantages. The form of the lunar series expansion is mathematically pleasing since it displays the influence of all harmonics compactly and elegantly, so that the dissection of the perturbations into short--periodic, long--periodic, and secular parts can be readily made and studied \citep{mL89}. The knowledge of these effects is essential for the determination of the stability of the orbits and the lifetimes of satellites. Of recent practical significance is the investigation of resonant effects on the inclined, nearly circular orbits of the navigation satellites for the management of the Global Navigation Satellite Systems
\citep[qq.v.][]{aR15,jD16}. Moreover, for the investigation of the geosphere and the interplanetary and interstellar space outside of it, highly--eccentric orbits (HEOs) with multi--day periods are being increasingly considered \citep{bE99, dD13}; such orbits are highly susceptible to the effects of lunar and solar perturbations. In order to understand the phase--space structure of MEOs and HEOs and to identify long--term stable regions as well as dynamical pathways that lead to slow, likely chaotic, variations in orbital elements, we must have an accurate representation of the lunar and solar disturbing function expansions.

The aim of this paper will be to discuss the problem anew and give the definitive expansions in the form of \citet{wK62} as in Section~\ref{sec:Kaula_development} and \citet{mL89} as in Section~\ref{sec:Lane_development}. Auxiliary formulae, as well as a considerable part of reckoning work for intermediate calculations, are presented in order to highlight and amend previous mistakes in the literature. An example in satellite dynamics is given in Section~\ref{sec:averages} to illustrate the utility of the results. We consider, in particular, orbits near the critical inclination $63.4^\circ$ in the region of semi--major axes where lunisolar (secular) perturbations generally become more significant than those of higher--order Earth gravity field distributions. Our analysis complements earlier investigations by \citet{mH81} and \citet{fDaM93}, and clarifies the essential role played by the regression of the lunar node in generating orbital chaos, as it was emphasized by \citet{aR15}
and demonstrated in \citet{jD16}. Some conclusions are drawn in Section~\ref{sec:conclusions}.

\section{Kaula's development of the disturbing function} \label{sec:Kaula_development}

The purpose of this Section is to recall the expansions of the solar and lunar disturbing functions for an artificial satellite (or a space debris) in terms of the equatorial elements of both the perturbed and perturbing bodies. Such development of the gravitational effects of the Sun and Moon on a close--Earth object was derived by \citet{wK62} with the intention of including the lunisolar perturbations as well as the perturbations due to the geopotential \citep[see][]{wK66} in the equations of variation of the orbital elements. Furthermore, due to its compact expression, the expansion also proves to be easy to implement in an algebraic manipulator and useful in evaluating various dynamical effects (e.g., secular resonances, the increase of eccentricity of satellite orbits, etc.)

Following \citet{wK62}, the gravitational potential due to a third--body perturber (either Sun or Moon),
expanded as a function of all orbital elements relative to the celestial equator, has the expression:
\beqa{Rsun}
\R^{*}&=&\G m^{*}\sum_{l=2}^{\infty}\sum_{m=0}^l \sum_{p=0}^l \sum_{h=0}^l \sum_{q=-\infty}^\infty \sum_{j=-\infty}^\infty {a^l\over (a^{*})^{l+1}}
\ \epsilon_m\, {{(l-m)!}\over {(l+m)!}}\nonumber\\
&&\times\ \F_{lmph}(I,I^*) H_{lpq}(e)\, G_{lhj}(e^{*})\ \cos(\varphi_{lmphqj})\ ,
\eeqa
where
\beqano
\F_{lmph}(I,I^*)&\equiv&F_{lmp}(I)\ F_{lmh}(I^{*})\ ,\nonumber\\
\varphi_{lmphqj}&\equiv& (l-2p)\omega+(l-2p+q)M-(l-2h)\omega^{*}-(l-2h +j)M^{*}+m(\Omega-\Omega^{*})
\eeqano
with $\G$ the gravitational constant; $m^*$ the mass of the disturbing body; $a$, $e$, $I$, $\omega$, $\Omega$ and $M$ the satellite's orbital elements\footnote{Using standard notation, $a$ denotes the semi--major axis, $e$ the eccentricity, $I$ the inclination, $\omega$ the argument of the perigee, $\Omega$ the longitude of the ascending node, and $M$ the mean anomaly.}; $a^*$, $e^*$, $I^*$, $\omega^*$, $\Omega^*$ and $M^*$ the corresponding elements of the  third--body perturber; the quantity $\epsilon_m$ is defined by
\beq{epsm}
    \epsilon_m = \left\{ \begin{array}{cl} 1 & \text{if } m = 0\ , \\ 2 & \text{if } m \in\integer\backslash\{0\}\ , \end{array} \right.
\eeq
the functions $H_{lpq}(e)$ and $G_{lhj}(e^{*})$ are the Hansen coefficients $X_{l-2p+q}^{l, l-2p}(e)$,
$X_{l-2h+j}^{-(l+1), l-2h}(e^*)$ (see Appendix~\ref{app:Hansen}); the terms $F_{lmp}(I)$ and $F_{lmh}(I^{*})$ are the Kaula's inclination functions (see Appendix~\ref{app:inclination_functions}).

Expression \eqref{Rsun} for the disturbing function may be compared with that given by \citet{cMsD99}. In this standard text on Solar system dynamics, the disturbing function for an inner secondary, which arises from the outer secondary mass perturbation potential, is similarly expanded in an infinite series in the osculating elements referred to the equator of the primary
with $\R^*$ as in \equ{Rsun}, where
\beqa{eq:murray_dermott}
\F_{lmph}(I,I^*)&\equiv&\bar{F}_{lmp}(I)\ \bar{F}_{lmh}(I^{*})\ ,\nonumber\\
\varphi_{lmphqj}&\equiv& (l-2p+q)\lambda-q\varpi+(m-l+2p)\Omega-(l-2h+j)\lambda^{*}\nonumber\\
&+&j\varpi^{*}-(m-l+2h)\Omega^{*}\ ,
\eeqa
where $\lambda$ and $\lambda^{*}$ are the mean longitudes, $\varpi$ and $\varpi^{*}$ are the longitudes of pericenter, but in this case the quantities $\bar{F}_{lmp}(I)$ and $\bar{F}_{lmh}(I^{*})$ represent the modified Allan inclination functions, related to the corresponding quantities used by Kaula according to\footnote{The history of the inclination functions and a discussion of the vexing issues of notation are given by \citet{rGcW08}.}
$$
    \bar{F}_{lmp} (I)
    = \left\{ \begin{array}{cl} F_{lmp} (I), & (l - m) \text{ even}\ ,
        \\[1em] -i F_{lmp} (I), & (l - m) \text{ odd}\ . \end{array} \right.
$$
Thus, Allan's differ from Kaula's functions by a factor of $i = \sqrt{-1}$ when $l - m$ is odd. \citet{kEcM00} give the disturbing function in a similar form as \eqref{eq:murray_dermott}, again using the mean longitudes and longitude of pericenters as the angular quantities over the argument of pericenters and mean anomalies, but instead using the Kaula inclination functions.

The variation of the Sun's orbital elements with respect to the celestial equator
are well approximated by linear functions of time, therefore the expansion \eqref{Rsun}
can be successfully applied in modeling the solar perturbations. However, as far as the lunar disturbing function is concerned, as noted in various works \citep[see][]{gC62,yK66}, the Moon's inclination, node, and argument of perigee are not simple functions of time. In fact, given that the main perturbing effect is due to the Sun, the variation of the above mentioned lunar elements with respect to the celestial equator is nonlinear. In particular, the longitude of the ascending node varies between $-13^\circ$
and $13^\circ$ with a period of 18.6 years. Within the same interval, the inclination of the lunar orbit with respect to the celestial equator oscillates between $18.4^\circ$ and $28.6^\circ$. The change in the argument of perigee is also nonlinear.

On the contrary, if we consider the elements of the Moon with respect to the ecliptic plane, then the inclination is approximately constant, while the variations of the argument of perigee and the longitude of the ascending node are approximately linear. This remark motivates the introduction of a different approach. Precisely, following \citet{gG74}, \citet{sH80}, and \citet{mL89}, it is convenient to introduce a rotation of the spherical harmonics for the Moon, so that its orbital elements are referred to the ecliptic plane, while the orbital elements of the satellite (or space debris) remain unchanged, that is,
they are referred to the equatorial plane. This alternative approach is the content of Section~\ref{sec:Lane_development}.

\section{Giacaglia's and Lane's lunar disturbing function expansions} \label{sec:Lane_development}


Let us consider a reference frame centered in the Earth and a
material point (e.g., a satellite or space debris) orbiting around the Earth. The gravitational action
of a third body (e.g., the Moon or the Sun) provokes a potential
given by the disturbing function
\beq{eq:potential}
\mathcal{R}_k= \frac{\mathcal{G} m_k}{\rho_k} - \frac{\mathcal{G} m_k ({\bfm r}_k \cdot {\bfm r})}{r_k^3}\ ,
\eeq
where $m_k$ is the
mass of the third body, ${\bfm r}$ is the position vector of the
point mass, ${\bfm r}_k$ is the position vector of the third body,
and $\rho_k = \lvert {\bfm r} - {\bfm r}_k \rvert$. Expanding
\eqref{eq:potential} in Legendre polynomials, one obtains
$$
    \mathcal{R}_k
    \label{eq:legendre}
    = \frac{\mathcal{G} m_k}{r_k} \sum\limits_{l \geq 2}
        \left( \frac{r}{r_k} \right)^l P_l \left( \cos \psi_k \right)\ ,
$$
where $\psi_k$ is the ``geocentric elongation'' of the point mass
from the third body, i.e., $r_k\ r \cos \psi_k = {\bfm r}_k \cdot
{\bfm r}$.

Assuming that the Moon is the third body, as mentioned before, it is convenient to
express the position of the Moon in the \sl ecliptic \rm frame, so that
the inclination $I_k$ becomes nearly constant, while the argument of
perihelion $\omega_k$, the longitude of the ascending node
$\Omega_k$, and the mean anomaly $M_k$ vary almost linearly with
time, with rates respectively equal to $0.164^\circ$/day,
$-0.053^\circ$/day, and $13.06^\circ$/day. This remark suggests that it is convenient
to express the elements of the point mass with respect to the
celestial equator and the elements of the Moon with respect to the
ecliptic plane \citep[compare with][]{yK66,gG74,mL89,sH80}.

\vskip.1in


The aim of this Section is to prove that the disturbing function $\mathcal{R}_k$ in \equ{eq:potential} can be expanded as in Proposition~\ref{pro:Lane} below, where the elements $I_k$, $M_k$, $\omega_k$, $\Omega_k$ of the Moon are referred to the ecliptic frame. The following expansion of the potential induced by the Moon
follows closely \citet{mL89}.

To this end, we premise the following result. Let $P_l^m(\cdot)$ be the associated Legendre functions of
degree $l$ and order $m$; let $a$, $a_k$ be the semi--major axes of the point mass and
the Moon, respectively; let $(\alpha,\delta)$, $(\alpha^\prime,\delta^\prime)$ be the
right ascension and declination with respect to the equator of the point mass and the Moon, respectively.
Then, using the spherical harmonic addition theorem, the expansion \equ{eq:potential} can be written as
\begin{align}
    \mathcal{R}_k
    \label{eq:legendre_add}
    = \frac{\mathcal{G} m_k}{a_k} \sum\limits_{l \geq 2} \sum\limits_{m=0}^l
        \frac{\epsilon_m (l - m)!}{(l + m)!} \left( \frac{a}{a_k} \right)^l \left( \frac{r}{a} \right)^l
        \left( \frac{a_k}{r_k} \right)^{l+1} P_l^m \left( \sin \delta \right) P_l^m \left( \sin \delta^\prime \right)
        \cos \left( m (\alpha - \alpha^\prime) \right)\ ,
\end{align}
where the quantity $\epsilon_m$ is defined by the relation \eqref{epsm}.

\begin{proposition}\label{pro:Lane}
Let $F_{lmp}(I)$ and $F_{lsq} (I_k)$ be the Kaula's inclination functions (see Appendix~\ref{app:inclination_functions})
with the inclination $I$ referred to the celestial equator and the inclination $I_k$ referred to the ecliptic, and let $X_r^{n,m}(e)$ denote the Hansen coefficients (see Appendix~\ref{app:Hansen}). Let the quantities $\epsilon_m$ be defined as in
\equ{epsm}.

Let us introduce the quantities $\bar{\theta}_{lmpj}$, $\bar{\theta}_{lsqr}^\prime$ as
\begin{align}
\label{eq:Theta_bar}
\begin{array}{rcl}
    \bar{\theta}_{lmpj} & = & (l - 2p) \omega + (l - 2p + j) M + m \Omega\ , \\[0.5em]
    \bar{\theta}_{lsqr}^\prime & = & (l - 2q) \omega_k + (l - 2q + r) M_k + s (\Omega_k - \pi/2)\\[0.5em]
\end{array}
\end{align}
and let the functions $U_l^{m,s}$ be defined as

\beqa{ULane}
U_l^{m,s}&=&\sum_{r =\max(0,-(m+s))}^{\min(l-s,l-m)} (-1)^{l-m-r}
\left(\begin{array}{c}
  l+m \\
  m+s+r \\
\end{array}\right)\
\left(\begin{array}{c}
  l-m \\
  r \\
\end{array}\right)\
\cos^{m+s+2r}({\varepsilon\over 2})\sin^{-m-s+2(l-r)}({\varepsilon\over 2})\ ,\nonumber\\
\eeqa
where $\epsilon$ is the ecliptic inclination.

Then, we have that the potential $\R_k$ in \equ{eq:potential} can be expanded as
\beqa{Rgood}
    \mathcal{R}_k
    \nonumber
    & = &\sum\limits_{l \geq 2} \sum\limits_{m = 0}^l \sum\limits_{p = 0}^l \sum\limits_{s = 0}^l
        \sum\limits_{q = 0}^l \sum\limits_{j = -\infty}^{+\infty} \sum\limits_{r = -\infty}^{+\infty}
        (-1)^{m+s}\ (-1)^{k_1} \frac{\mathcal{G} m_k \epsilon_m \epsilon_s}{2 a_k} \frac{(l - s)!}{(l + m)!}
        \left( \frac{a}{a_k} \right)^l \\
    \nonumber
    &\times& F_{lmp} (I) F_{lsq} (I_k)
        H_{lpj} (e) G_{lqr} (e_k) \\
    \label{eq:lane}
    &\times& \left\{ (-1)^{k_2} U_l^{m, -s}
        \cos \left( \bar{\theta}_{lmpj} + \bar{\theta}_{lsqr}^\prime - y_s \pi \right)
        + (-1)^{k_3} U_l^{m, s} \cos \left( \bar{\theta}_{lmpj} - \bar{\theta}_{lsqr}^\prime - y_s \pi \right)  \right\}\ ,\nonumber\\
\eeqa
where $y_s=0$ for $s$ even and $y_s=1/2$ when $s$ is odd, $k_1=[m/2]$, $k_2 = t (m + s - 1) + 1$, $k_3 = t (m + s)$
with $t=(l-1)$ mod 2.
\end{proposition}

The proof of Proposition~\ref{pro:Lane} will be given at the end of this Section; we need
first some auxiliary results.

\begin{remark}
$(i)$ Notice that an alternative expression of the functions $U_l^{m,s}$ defined in \eqref{ULane} is the following
(compare with \citet{gG74}):
\begin{align}
    U_l^{m,s}
    \label{eq:U_case_greater}
    & = (-1)^{l - m} \left( \begin{array}{c} l + m \\ l - s \end{array} \right)
        \left( \cos \frac{\epsilon}{2} \right)^{m + s} \left( \sin \frac{\epsilon}{2} \right)^{s - m}
        F \left( -l + s, l + s + 1, m + s + 1; \cos^2 \frac{\epsilon}{2} \right)
\intertext{for $m + s \geq 0$, and}
    U_l^{m,s}
    \label{eq:U_case_less}
    & = (-1)^{l - s} \left( \begin{array}{c} l - m \\ l + s \end{array} \right)
        \left( \cos \frac{\epsilon}{2} \right)^{-m - s} \left( \sin \frac{\epsilon}{2} \right)^{m - s}
        F \left( -l - s, l - s + 1, -m - s + 1; \cos^2 \frac{\epsilon}{2} \right)
\end{align}
for $m + s < 0$, where the hypergeometric series $F = \tensor[_2]{F}{_1}$ is
defined by
\begin{align*}
    F (a, b, c; x) = \sum\limits_{n=0}^{+\infty} \frac{(a)_n (b)_n}{(c)_n} \frac{x^n}{n!}\ ,
\end{align*}
in which we used the Pochhammer symbol $(a)_0 = 1, (a)_n = a (a + 1)
\ldots (a + n - 1)$.
We observe that the expressions \equ{eq:U_case_greater},
\equ{eq:U_case_less} of the functions $U_l^{m,s}$ in terms of the hypergeometric
series are attributed to Jacobi in \citet{CH}.

$(ii)$ The final expansion for $\mathcal{R}_k$ in \equ{Rgood}
coincides with that given by \citet{mL89}, although the proof in \citet{mL89} is not completely correct.
To mention one point, the relation \equ{eq:sphere_harmon_rot} of Lemma~\ref{lem:plm} below appears
at p. 290 of \citet{mL89}; however, the functions $\Lambda_l^{m,s}$ appearing in \equ{lambda}
below are not properly defined in \citet{mL89}, unless one takes the definition \equ{ULane}
for the functions $U_l^{m,s}$, which in \citet{mL89} are multiplied by the factor $(-1)^{m-s}$.
\end{remark}

\vskip.1in

We introduce the quantities $C_l^m$, $S_l^m$ defined as
\begin{align*}
    C_l^m \equiv A_l^m \cos m \alpha^\prime, \quad S_l^m \equiv A_l^m \sin m \alpha^\prime\ ,
\end{align*}
where $A_l^m$ is given by
\begin{align*}
    A_l^m \equiv \frac{\mathcal{G} m_k \epsilon_m (l - m)!}{a_k (l + m)!} \left( \frac{a}{a_k} \right)^l
        \left( \frac{r}{a} \right)^l \left( \frac{a_k}{r_k} \right)^{l + 1} P_l^m \left( \sin \delta^\prime \right)\ .
\end{align*}
With this setting we can write \eqref{eq:legendre_add} as
\begin{align}
    \mathcal{R}_k
    \label{eq:kaula_tmp}
    = \sum\limits_{l \geq 2} \sum\limits_{m=0}^l P_l^m \left( \sin \delta \right)
        \left( C_l^m \cos m \alpha + S_l^m \sin m \alpha \right)\ .
\end{align}
According to \citet[][p. 31, equation (3.53) with $P_{lm}$ replaced by $P_l^m$ and
p. 34, equation (3.61)]{wK66},\footnote{Notice that $P_l^m(\sin\delta)=(-1)^m\,P_{lm}(\sin\delta)$.} one can rewrite \eqref{eq:kaula_tmp} as
\begin{align}
    \mathcal{R}_k
    \nonumber
    & = \sum\limits_{l \geq 2} \sum\limits_{m=0}^l \sum\limits_{p=0}^l (-1)^m\ F_{lmp} (I)
        \left\{ \left[ \begin{array}{c} C_l^m \\[0.5em] -S_l^m
        \end{array} \right]_{l - m \text{ odd}}^{l - m \text{ even}}
        \cos \left( (l - 2p) (\omega + f) + m \Omega \right) \right. \\
    \label{eq:kaula}
    & \hspace{12pt} + \left. \left[ \begin{array}{c} S_l^m \\[0.5em] C_l^m
        \end{array} \right]_{l - m \text{ odd}}^{l - m \text{ even}}
        \sin \left( (l - 2p) (\omega + f) + m \Omega \right) \right\}\ ,
\end{align}
where $f$, $\omega$, $\Omega$ are, respectively, the true anomaly,
the argument of perigee, the longitude of the ascending node of
the point mass, while the function $F_{lmp} (I)$ is the
Kaula's inclination function defined in Appendix~\ref{app:inclination_functions}.

\vskip.1in

For short, we denote by
$$
    \theta_{lmp} \equiv (l - 2p) (\omega + f) + m \Omega\ .
$$
Since we aim to have the elements of the Moon with respect to the
ecliptic plane, we need to transform the spherical harmonics
through a rotation.
To this end, we recall a result due to \citet{bJ65},
which is of crucial importance for the proof of the main result, Proposition~\ref{pro:Lane}.

\begin{lemma}\label{lem:plm}
Let $(\delta_k, \alpha_k)$ be the
ecliptic latitude and longitude\footnote{In \citet{mL89}
$\delta_k$, $\alpha_k$ denote the ecliptic declination and right
ascension of the Moon.} of the Moon. Let us introduce the quantities
$\Lambda_l^{m,s}$ defined as
\beq{lambda}
    \Lambda_l^{m,s}
    = \frac{(l - s)!}{(l - m)!} e^{i (m - s) \pi/2} U_l^{m,s}\ .
\eeq
Then, we have the following relation:
\beq{eq:sphere_harmon_rot}
    P_l^m \left( \sin \delta^\prime \right) e^{i m \alpha^\prime}
    = \sum\limits_{s = -l}^l \Lambda_l^{m,s} P_l^s (\sin \delta_k) e^{i s \alpha_k}\ .
\eeq
\end{lemma}

\begin{proof}
We want to express $P_l^m(\sin \delta^\prime) e^{i m \alpha^\prime}$ in terms of
$(\delta_k, \alpha_k)$ by using the relations
\begin{align}
    \label{eq:delta_alpha_relations}
    \begin{array}{rcl} \cos \delta^\prime e^{i \alpha^\prime} & = & \cos \delta_k \cos \alpha_k
        + i \left( \cos \delta_k \sin \alpha_k \cos \epsilon - \sin \delta_k \sin \epsilon \right) \\[0.5em]
    \sin \delta^\prime & = & \cos \delta_k \sin \alpha_k \sin \epsilon
        + \sin \delta_k \cos \epsilon\ . \end{array}
\end{align}
Using the properties of spherical harmonics under a rotation,
let us write the quantity $P_l^m \left( \sin \delta^\prime \right) e^{i m \alpha^\prime}$ as
$$
    P_l^m \left( \sin \delta^\prime \right) e^{i m \alpha^\prime}
    = \sum\limits_{s = -l}^l \Lambda_l^{m,s} P_l^s (\sin \delta_k) e^{i s \alpha_k}\ ,
$$
where $\Lambda_l^{m,s} = \Lambda_l^{m,s} (\epsilon)$
are suitable  functions that depend on
the ecliptic inclination $\epsilon$. According to
\citet{gG74}, the expression of $\Lambda_l^{m,s}$ is obtained as
follows. From the previous expression  and using the
orthogonality condition of the functions $P_l^m$, one finds that
\begin{align}
    \Lambda_l^{m,s}
    \label{eq:Lambda_orthog}
    = \frac{\epsilon_s (2 l + 1)}{4 \pi} \frac{(l - s)!}{(l + s)!}
        \int\limits_{-\pi/2}^{\pi/2} \cos \delta_k\, \mathrm{d} \delta_k
        \int\limits_0^{2 \pi} P_l^m (\sin \delta^\prime) e^{i m \alpha^\prime}
        P_l^s (\sin \delta_k) e^{i s \alpha_k}\, \mathrm{d} \alpha_k\ .
\end{align}
Inserting \eqref{eq:delta_alpha_relations} in
\eqref{eq:Lambda_orthog} and making the integral one obtains that the functions $\Lambda_l^{m,s}$
are given as in \equ{lambda}.
\end{proof}

We now need another auxiliary result to transform the term $C_l^m\cos m\alpha+S_l^m\sin m\alpha$ in
\equ{eq:kaula_tmp}.

\begin{lemma}{[\citet[][p. 290]{mL89}]}\label{lem:CS}
Define the quantities $C_l^{m,s}$, $S_l^{m,s}$, $A_l^{m,s}$
as
\begin{align}
\label{eq:CSA}
\begin{array}{rcl}
    C_l^{m,s} & \equiv & \frac{1}{2} \left( U_l^{m,s} + (-1)^s U_l^{m,-s} \right), \\[0.5em]
    S_l^{m,s} & \equiv & \frac{1}{2} \left( U_l^{m,s} - (-1)^s U_l^{m,-s} \right), \\[0.5em]
    A_l^{m,s} & \equiv & \ds \frac{\mathcal{G} m_k \epsilon_m \epsilon_s}{a_k} \frac{1}{(l + m)!}
        \left( \frac{a}{a_k} \right)^l \left( \frac{r}{a} \right)^l \left( \frac{a_k}{r_k} \right)^{l + 1}\ ;
\end{array}
\end{align}
then, we have:
\begin{align}
    C_l^m + i S_l^m
    \label{eq:C_imag_S}
    = i^m \sum\limits_{s=0}^l A_l^{m,s} P_l^s (\sin \delta_k) (l - s)!
        \left\{ C_l^{m,s} \cos \left( s (\alpha_k - \pi/2) \right)
        + i S_l^{m,s} \sin \left( s (\alpha_k - \pi/2) \right) \right\}\ .
\end{align}
\end{lemma}

The proof of Lemma~\ref{lem:CS} is detailed in Appendix~\ref{app:lemCS}.

\begin{remark} 
As first noted by \citet{mL89}, in \citet{gG74} there was an incorrect derivation of the above formula
for $C_l^m+iS_l^m$, which propagated over the paper.
\end{remark} 

As in \citet{mL89}, we use Lemma~\ref{lem:CS} to transform the expansion \equ{eq:kaula}.
However, we remark that the expression \equ{eq:expan_Theta} below differs from formula (6)
in \citet[][p. 291]{mL89} by a factor $(-1)^{m+s}$. However, although the formulation of
Lemma~\ref{lem:plm} in \citet{mL89} is not correct, the functions $U_l^{m,s}$ in \citet{mL89} contain
a factor $(-1)^{m-s}$, which compensates the factor $(-1)^{m+s}$ in the expansion \equ{eq:expan_Theta}.

\begin{lemma}{[\citet[][p. 291-292]{mL89}]}\label{lem:theta}
Let $I_k$ be the inclination of the Moon, referred
to the ecliptic plane. Let
\begin{align*}
    \theta_{lsq}^\prime \equiv (l - 2q) (\omega_k + f_k) + s (\Omega - {\pi\over 2})\ ,
\end{align*}
where $f_k$ is the true anomaly of the Moon referred to the ecliptic.

Then, we have the following expansion:
\begin{align}
    \mathcal{R}_k
    \label{eq:expan_Theta}
    = \sum\limits_{l \geq 2} \sum\limits_{m = 0}^l \sum\limits_{p = 0}^l \sum\limits_{s = 0}^l
        \sum\limits_{q = 0}^l (-1)^{m+s}\ (-1)^{k_1} A_l^{m,s} (l - s)! F_{lmp} (I) F_{lsq} (I_k) \Theta_{lmpsq}\ ,
\end{align}
where, if $m$ is even, $l - m$ is even:
\begin{align}
    \label{eq:Theta_even}
    \Theta_{lmpsq} = \left\{
    \begin{array}{cl}
        \frac{1}{2} \left[ (-1)^s U_l^{m,-s} \cos \left( \theta_{lmp} + \theta_{lsq}^\prime \right)
        + U_l^{m,s} \cos \left( \theta_{lmp} - \theta_{lsq}^\prime \right) \right] & \text{$l - s$ even} \\[0.5em]
        \frac{1}{2} \left[ (-1)^s U_l^{m,-s} \sin \left( \theta_{lmp} + \theta_{lsq}^\prime \right)
        - U_l^{m,s} \sin \left( \theta_{lmp} - \theta_{lsq}^\prime \right) \right] & \text{$l - s$ odd}\ ;
    \end{array} \right.
\end{align}
if $m$ is even, $l - m$ is odd; or $m$ is odd, $l - m$ is even:
\begin{align}
    \label{eq:Theta_even_odd}
    \Theta_{lmpsq} = \left\{
    \begin{array}{cl}
        \frac{1}{2} \left[ (-1)^s U_l^{m,-s} \sin \left( \theta_{lmp} + \theta_{lsq}^\prime \right)
        + U_l^{m,s} \sin \left( \theta_{lmp} - \theta_{lsq}^\prime \right) \right] & \text{$l - s$ even} \\[0.5em]
        \frac{1}{2} \left[ -(-1)^s U_l^{m,-s} \cos \left( \theta_{lmp} + \theta_{lsq}^\prime \right)
        + U_l^{m,s} \cos \left( \theta_{lmp} - \theta_{lsq}^\prime \right) \right] & \text{$l - s$ odd}\ ;
    \end{array} \right.
\end{align}
if $m$ is odd, $l - m$ is odd:
\begin{align}
    \label{eq:Theta_odd}
    \Theta_{lmpsq} = \left\{
    \begin{array}{cl}
        \frac{1}{2} \left[ -(-1)^s U_l^{m,-s} \cos \left( \theta_{lmp} + \theta_{lsq}^\prime \right)
        - U_l^{m,s} \cos \left( \theta_{lmp} - \theta_{lsq}^\prime \right) \right] & \text{$l - s$ even} \\[0.5em]
        \frac{1}{2} \left[ -(-1)^s U_l^{m,-s} \sin \left( \theta_{lmp} + \theta_{lsq}^\prime \right)
        + U_l^{m,s} \sin \left( \theta_{lmp} - \theta_{lsq}^\prime \right) \right] & \text{$l - s$ odd}\ .
    \end{array} \right.
\end{align}
\end{lemma}

As mentioned in \citet{mL89}, the proof of Lemma~\ref{lem:theta} requires ``a considerable amount
of tedious algebra". Although Lemma~\ref{lem:theta} is essential to get the correct expression for $\R_k$,
for the readability of this paper we postpone its proof to Appendix~\ref{app:Lemma7}.

Using the relation $\cos (x - {\pi\over 2}) = \sin x$, we can write
\eqref{eq:Theta_even}-\eqref{eq:Theta_odd} in a single case as follows.

\begin{lemma}{[\citet[p. 292]{mL89}]}\label{lem:theta2}
Let the quantity $y_s$ be defined as $y_s = s/2 - [{s/2}]$.
Then, we can write:
\begin{align}
    \label{eq:Theta_unite}
    \Theta_{lmpsq} = \frac{1}{2} \left[ (-1)^{k_2} U_l^{m, -s} \cos \left( \theta_{lmp}
        + \theta_{lsq}^\prime - y_s \pi \right) + (-1)^{k_3} U_l^{m, s} \cos \left( \theta_{lmp}
        - \theta_{lsq}^\prime - y_s \pi \right) \right],
\end{align}
where $k_2 = t (m + s - 1) + 1$, $k_3 = t (m + s)$, $t = (l - 1)\ \mod 2$ (i.e., $t = 0$ if $l - 1$ is even, $t = 1$, if $l - 1$ is
odd).
\end{lemma}

The proof of Lemma~\ref{lem:theta2} is a check of the different cases
in which $m$ is even or odd, $l$ is even or odd. The proof is detailed in Appendix~\ref{app:Lemma8}.\\

We are finally ready to give the proof of Proposition~\ref{pro:Lane}, which is based on
Lemma~\ref{lem:CS}, Lemma~\ref{lem:theta}, and Lemma~\ref{lem:theta2}.

\begin{proof}[Proof of Proposition~\ref{pro:Lane}]
We can write $\R_k$ in \eqref{eq:expan_Theta} with $A_l^{m,s}$ as in Lemma~\ref{lem:CS} (see \equ{eq:CSA}) and
$\Theta_{lmpsq}$ as in Lemma~\ref{lem:theta2} (see \equ{eq:Theta_unite}). This leads to the following expression:
\begin{align}
    \mathcal{R}_k
    \nonumber
    & = \sum\limits_{l \geq 2} \sum\limits_{m = 0}^l \sum\limits_{p = 0}^l \sum\limits_{s = 0}^l
        \sum\limits_{q = 0}^l (-1)^{m+s}\ (-1)^{k_1} \frac{\mathcal{G} m_k \epsilon_m \epsilon_s}{2 a_k}
        \frac{(l - s)!}{(l + m)!} \left( \frac{a}{a_k} \right)^l \left( \frac{r}{a} \right)^l \left( \frac{a_k}{r_k} \right)^{l + 1}
        F_{lmp} (I) F_{lsq} (I_k) \\
    \label{eq:lane_temp}
    & \hspace{12pt} \times \left\{ (-1)^{k_2} U_l^{m, -s}
        \cos \left( \theta_{lmp} + \theta_{lsq}^\prime - y_s \pi \right)
        + (-1)^{k_3} U_l^{m, s} \cos \left( \theta_{lmp} - \theta_{lsq}^\prime - y_s \pi \right)  \right\}\ .
\end{align}
Next, we use the following expansion in terms of Hansen's coefficients $X_r^{n, m} (e)$ \citep{jC72,gG76}:
\begin{align}
    \label{eq:Hansen_coeff}
    \left( \frac{r}{a} \right)^l e^{i m f} = \sum\limits_{j = -\infty}^{+\infty} X_{m + j}^{l, m} (e) e^{i (m + j) M}\ .
\end{align}
Using \eqref{eq:Hansen_coeff} and \eqref{eq:Theta_bar}, we can
write \eqref{eq:lane_temp} as in \equ{eq:lane}, due to \eqref{eq:Hansen_coeff}
with $m = l - 2p$ and to
\begin{align*}
    \left( \frac{a_k}{r_k} \right)^{l + 1} e^{i m f_k}
    = \sum\limits_{r = -\infty}^{+\infty} X_{m + r}^{-(l + 1), m} (e_k) e^{i (m + r) M_k}\ ,
\end{align*}
with $m = l - 2q$.
Recalling the relation between the Hansen's coefficients and the functions $H_{lpj}$, $G_{lqr}$
(see Appendix~\ref{app:Hansen}), one is led to the expansion \equ{eq:lane}.
\end{proof}

In Appendix~\ref{app:orbit_comparison}, we validate the expansions \eqref{Rsun} and \equ{eq:lane}
by comparing some orbits propagated both by using a Cartesian model and a model based on the above lunisolar expansions.

\section{An application to lunisolar resonances: the critical inclination secular resonance} \label{sec:averages}

The expansions \eqref{Rsun} and  \equ{eq:lane} allow one to have very versatile formulas for the solar and lunar potentials, which may be used in various investigations. As an example, we describe here an application of these expansions to the study of lunisolar secular resonances. Truncating the series \eqref{Rsun} and \equ{eq:lane} to second order in the ratio of semi--major axes and averaging over both mean anomalies of the point mass and of the third body, we show how the solar and lunar disturbing functions may be used to get a global picture of the long--term complex evolution of resonant orbits. In particular, we present some results obtained for the so--called critical inclination resonance, which arises when the orbital inclination of the point mass is equal to $63.4^\circ$ \citep[see][]{sH80,mH81,fDaM93,ElyHowell,aR15}. Three sample cases are considered, namely Molniya 1-81, Molniya 1-88 and Molniya 1-86, and their dynamics is investigated by evaluating the Fast Lyapunov Indicators (hereafter FLIs),
introduced in \cite{froes} and used in similar contexts in \cite{CGmajor}, \cite{CGext}, \cite{CGminor}, and \citet{jD16}.

Being interested in the long--term dynamics of resonant orbits, the short--periodic terms that depend on the  mean anomaly of the satellite or the mean anomaly of the perturbing body can be averaged over from the disturbing functions. Thus, we consider a Hamiltonian of the form
\begin{equation}\label{Hamiltonian}
\H=\H_{Kep}+\H_{Geo}+\H_{Moon}+\H_{Sun}\ ,
\end{equation}
where $\H_{Kep}$ represents the Kepler Hamiltonian, while the functions $\H_{Geo}$, $\H_{Moon}$ and $\H_{Sun}$ describe the perturbations due to Earth, Moon and Sun, respectively, averaged over the mean anomalies of the satellite and the perturbing body.

Using the Delaunay action--angle variables $(L,G,H,M,\omega,\Omega)$, where the actions are defined by
\begin{equation}\label{Delaunay}
L=\sqrt{\mu_E a}\,, \quad G=L\sqrt{1-e^2}\,,\quad H=G \cos I
\end{equation}
with $\mu_E= \G m_E$ the product of the gravitational constant $\G$ and the Earth's mass $m_E$, then the Keplerian part is given by
$$ \H_{Kep}(L)=-\frac{\mu_E^2}{2 L^2}\,.  $$
Concerning the disturbing function due to the Earth, we consider only the most important contribution, corresponding to the $J_2$ gravity coefficient of the secular part \citep[see][]{CGmajor}, precisely
$$
\H_{geo}(L,G,H)={{R_E^2 J_2 \mu_E^4}\over {4}}\ {{1}\over {L^3G^3}}\ (1-3{H^2\over G^2})\ ,
$$
where $R_E$ is the mean equatorial radius of the Earth.
For the effects of the second power of $J_2$ and of the higher--order harmonics we refer to \cite{CDD1994}.
Our neglect of  the second--order $J_2$ contribution means that our quantitative
results at lower semi--major axes, where these effects become more important, should be taken with a grain of salt.

The perturbations due to Moon and Sun are given by
$$\H_{Moon}=-\overline{\R}_{Moon}\,,\qquad \H_{Sun}=-\overline{\R}_{Sun}\ ,$$
where $\overline{\R}_{Moon}$ and $\overline{\R}_{Sun}$ are obtained from \equ{eq:lane} and  \eqref{Rsun}, respectively, by setting $l=2$ in the two expansions, so that the lunar and solar potentials are approximated by quadrupole fields, and by taking the average over the mean anomalies of both the point mass and the perturbing body. Of course, in the resulting expansions, the orbital elements are expressed in terms of the Delaunay variables.

Since $M$ is an ignorable variable, its conjugated action $L$ (or equivalently the semi--major axis $a$)  is a constant. Thus, the Hamiltonian system described by \eqref{Hamiltonian} is non--autonomous with two degrees of freedom.
Analytical studies dating back to the '60s, see, for instance \citet{pM61} and \citet{Harr1969},
have shown that the expansion up to the order $l=2$ of the gravitational potential due to a third--body
perturber in the secular problem is independent of the perturber's argument of periapsis.
And, in fact, the existence of the fundamental Lidov-Kozai cycles hinges on this fact \citep[q.v.,][]{yLsN11}.
The following Proposition ~\ref{prop:omega_dependence} confirms this result and
shows rigorously that the Hamiltonian $\H$ depends periodically on time just through
the longitude of the lunar ascending node $\Omega_k$.

\begin{proposition} \label{prop:omega_dependence}
The functions $\overline{\R}_{Moon}$ and $\overline{\R}_{Sun}$ do not depend on the argument of perigee of the third body.
\end{proposition}
\begin{proof}  It is enough to prove the above statement just for one expansion, let us say for $\overline{\R}_{Moon}$. In the other case the argument is the same.

Since $\overline{\R}_{Moon}$ is obtained from \equ{eq:lane} by taking $l=2$ and averaging over the mean anomalies $M$ and $M_k$, it contains just terms for which $2-2 p+j=0$ and $2-2 q+r=0$, where $p$ and $q$ take the values $0,1,2$ and $r,j
 \in \mathbb{Z}$. Thus, to prove the statement we have to show that these terms do not depend on $\omega_k$.

In fact, we have to discuss  three cases: $q=0$, $q=1$ and $q=2$. If $q=0$, then from the equation $2-2 q+r=0$, it follows that $r=-2$ and, as a consequence, the Hansen coefficient $G_{lqr}(e_k)$ associated to the terms for which $l=2$, $q=0$, $r=-2$ is $X^{-(l+1),l-2q}_{l-2q+r}(e_k)=X_0^{-3,2}(e_k)=0$. Therefore, all the terms having $l=2$, $q=0$, $r=-2$ are zero. Similarly, since $X_0^{-3,-2}(e_k)=0$, all terms of the expansions for which $l=2$, $q=2$, $r=2$ vanish. It remains to analyze the case $q=1$. From \eqref{eq:Theta_bar} it follows that $\bar{\theta}^\prime_{2s10}=s(\Omega_k- \pi/2)$ and thus, the terms for which $l=2$, $q=1$, $r=0$ do not depend on $\omega_k$.
\end{proof}

In view of the above result and of the fact that $\dot{\Omega}^*=0$, over timespans of interest, where $\Omega^*$ is the  longitude of the solar ascending node, it follows that $\H$ depends on time just through $\Omega_k$. Despite this simplification, the global dynamics of the system is quite complex.  Besides the fact that $\H$ is a non--autonomous, two degrees--of--freedom Hamiltonian, it also depends parametrically on the semi--major axis $a$ and the combined effects of $\H_{{G}eo}$, $\H_{Moon}$ and $\H_{Sun}$ lead to an intricate dynamics.
A major role in the long term evolution of orbital elements is played by the lunisolar secular resonances, which occur when some specific linear combinations of the secular precession frequencies vanish
\citep[see][and references therein, for a detailed presentation of the subject]{sH80, ElyHowell, aR15, jD16}.

\vskip.1in

Here we focus on the so--called critical inclination resonance, which occurs when the commensurability relation $\dot{\omega}=0$ holds.
Following \citet{sH80} \citep[see also][]{gC62,aR08}, this resonance together with other two types of secular resonances, characterized by the commensurability conditions $\dot{\Omega}=0 $ (polar resonance) and $\alpha \dot{\omega}+ \beta \dot{\Omega}=0$, with $\alpha, \beta \in \mathbb{Z}\setminus \{0\}$, form the class of resonances dependent only on the satellite's orbital inclination, called \sl inclination--dependent--only lunisolar resonances. \rm
The name is justified as long as one can approximate
$\dot \omega$, $\dot \Omega$ by the following well known formulae, which take into account only the effects of $J_2$ \citep[see][]{wK66, sH80}:
\beqa{omega12}
\dot\omega &\simeq& 4.98 \Bigl({R_E\over a}\Bigr)^{7\over 2}\ (1-e^2)^{-2}\ (5\cos^2 I-1)\ ^{\circ}/day\ ,\nonumber\\
\dot\Omega &\simeq& -9.97 \Bigl({R_E\over a}\Bigr)^{7\over 2}\ (1-e^2)^{-2}\ \cos I\ ^{\circ}/day\ ,
\eeqa
and, moreover, the variation of $\Omega_k$ is disregarded.  Indeed, from the first in $\eqref{omega12}$ the relation $\dot{\omega}=0$ holds for $I=63.4^\circ$ and for every value of $a$ and $e$. However, to be precise, since $\Omega_k$ varies periodically, a cosine argument of $\overline{\R}_{Moon}$ could depend also on  $\Omega_k$, and thus, besides $\dot{\omega}=0$, one also has the commensurability relations $2 \dot{\omega}+ s\dot{ \Omega}_k=0$ with $s=-2,-1,1,2$. One can say that each resonance of the above mentioned class, including the critical inclination one, splits into a multiplet of resonances. This splitting phenomenon is responsible for the existence of a very complex web--like background of resonances in the phase space, which leads to a chaotic variation of the orbital elements. An analytical estimate of the location of the resonance corresponding to each component of the multiplet, as a function of eccentricity and inclination, can be obtained by using $\eqref{omega12}$ (see, for example, Figure~2 in \citet{ElyHowell} or \citet{aR15}). Here, since we are using the Delaunay variables, we represent in Figure~\ref{WEB_structure} the web structure of resonances in the space of the actions $H$--$G$. To avoid confusions that might arise when we speak about a specific resonance, we will use the syntagma {\it exact resonance} when we refer to the component of the multiplet characterized by $s=0$ in \equ{eq:Theta_bar}, while the expression {\it whole resonance} means that we refer to all components of the multiplet.

We underline that the units
of length and time are normalized so that the geostationary
distance is unity (it amounts to $42\,164.17$ km) and that the period
of Earth's rotation is equal to $2 \pi$. As a consequence, from
Kepler's third law it follows that $\mu_E=1$. Therefore, unless
the units are explicitly specified, the action variables $L$, $G$ and $H$ are expressed in the above units.

Figure~\ref{WEB_structure} shows the structure of resonances for $a=13339.1$ km (top panels), $a=18851.7$ km (bottom left) and $a=26508.2$ km (bottom right). These values are not chosen by chance, but represent the semi--major axes of the satellites: Molniya 1-86, Molniya 1-88 and Molniya 1-81, respectively. The colored curves provide the location of the resonances, while the vertical black dashed line is drawn to point out the values of $H$ used in computing the FLI maps. In order to depict graphically the splitting phenomenon,  Figure~\ref{WEB_structure} top left panel shows the resonant structure for $G\in[0,G_{max}]$, where $G_{max}=\sqrt{\mu_E a}$. This plot contains also the horizontal black line $G=G_{min}$, where $G_{min}$ is computed from the condition that the distance of the perigee cannot be smaller than the radius of the Earth, that is
$$G_{min}=\sqrt{\frac{(2 a- R_E)\mu_E R_E}{a}}\ .$$
Therefore, the interval of interest is $[G_{min}, G_{max}]$ and the top right panel of Figure~\ref{WEB_structure} magnifies the region associated to the orbits that do not collide with the Earth. The bottom plots of Figure~\ref{WEB_structure}  are also obtained for $G\in [G_{min}, G_{max}]$.

\begin{figure}[h]
\centering
\vglue0.1cm
\hglue0.1cm
\includegraphics[width=7.3truecm,height=6truecm]{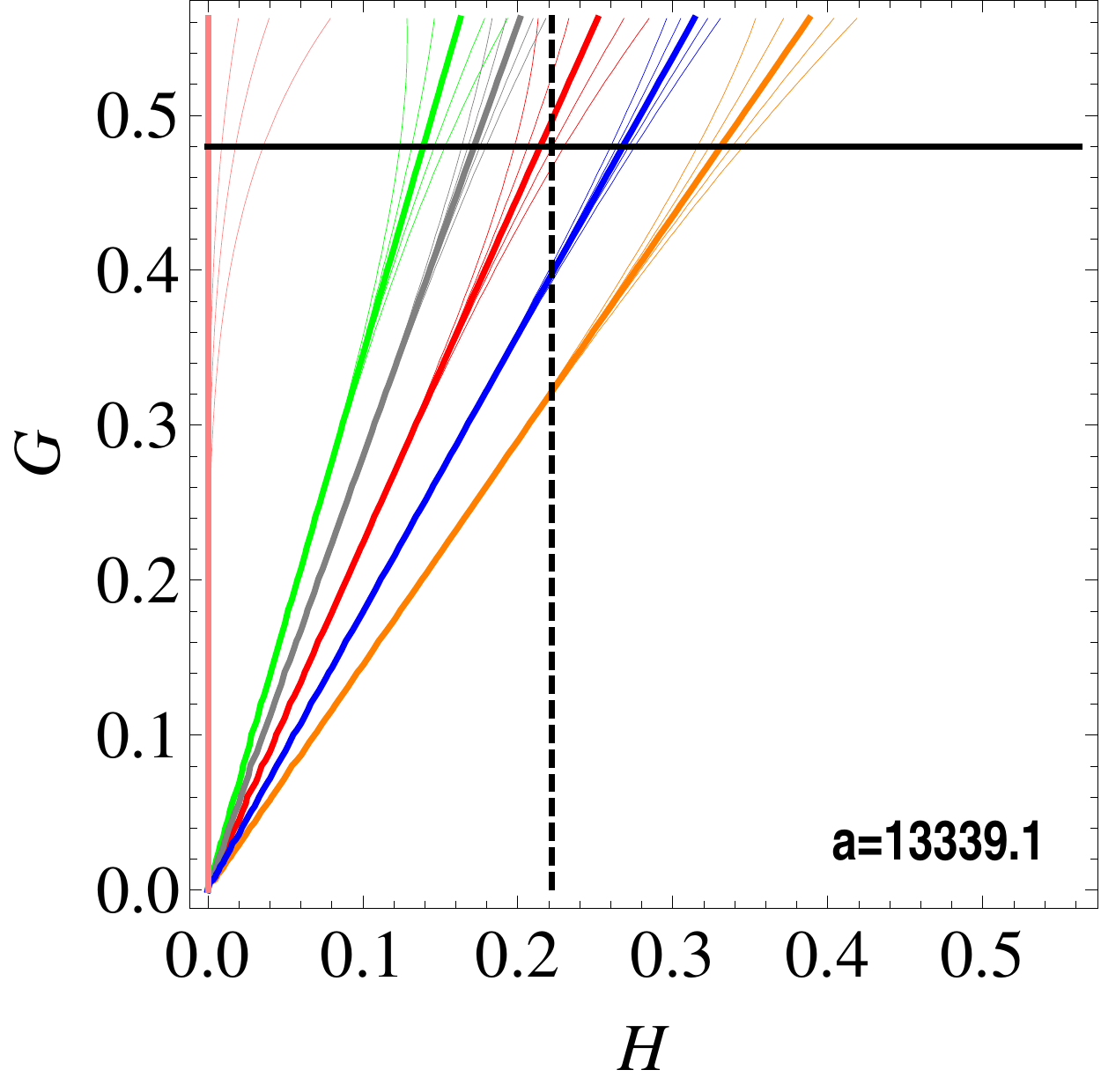}
\includegraphics[width=7.3truecm,height=6truecm]{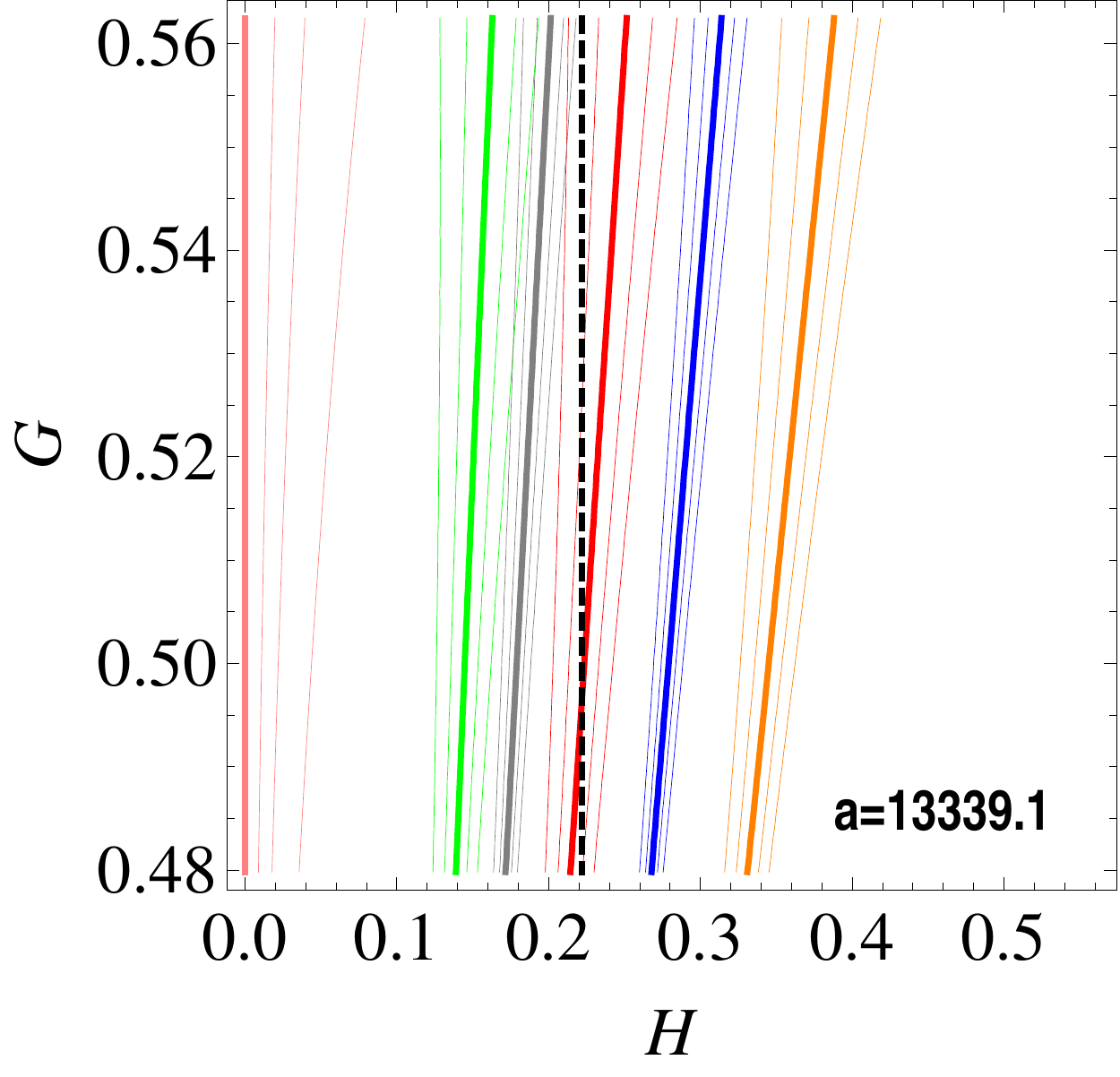}
\includegraphics[width=7.5truecm,height=6truecm]{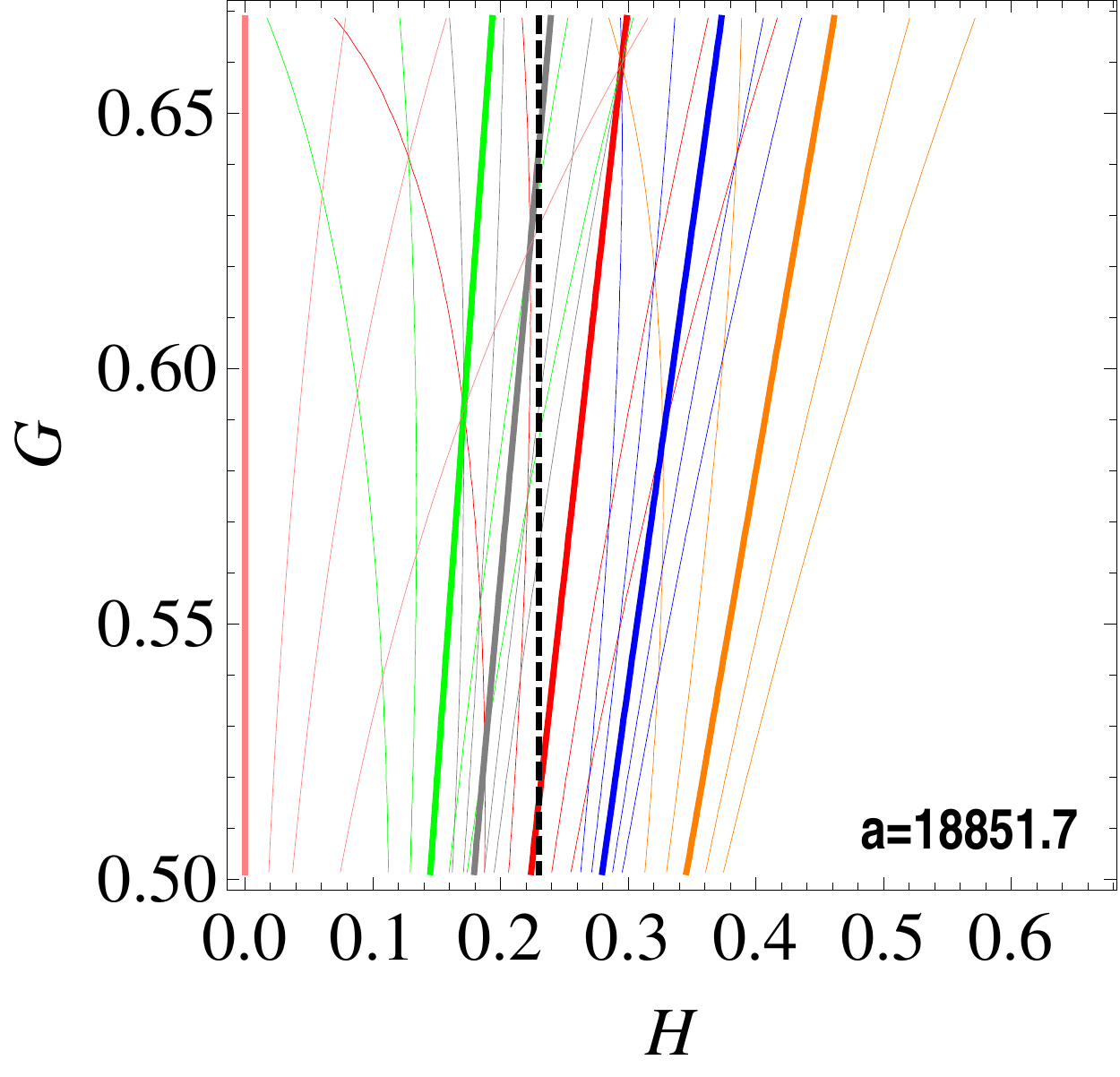}
\includegraphics[width=7.5truecm,height=6truecm]{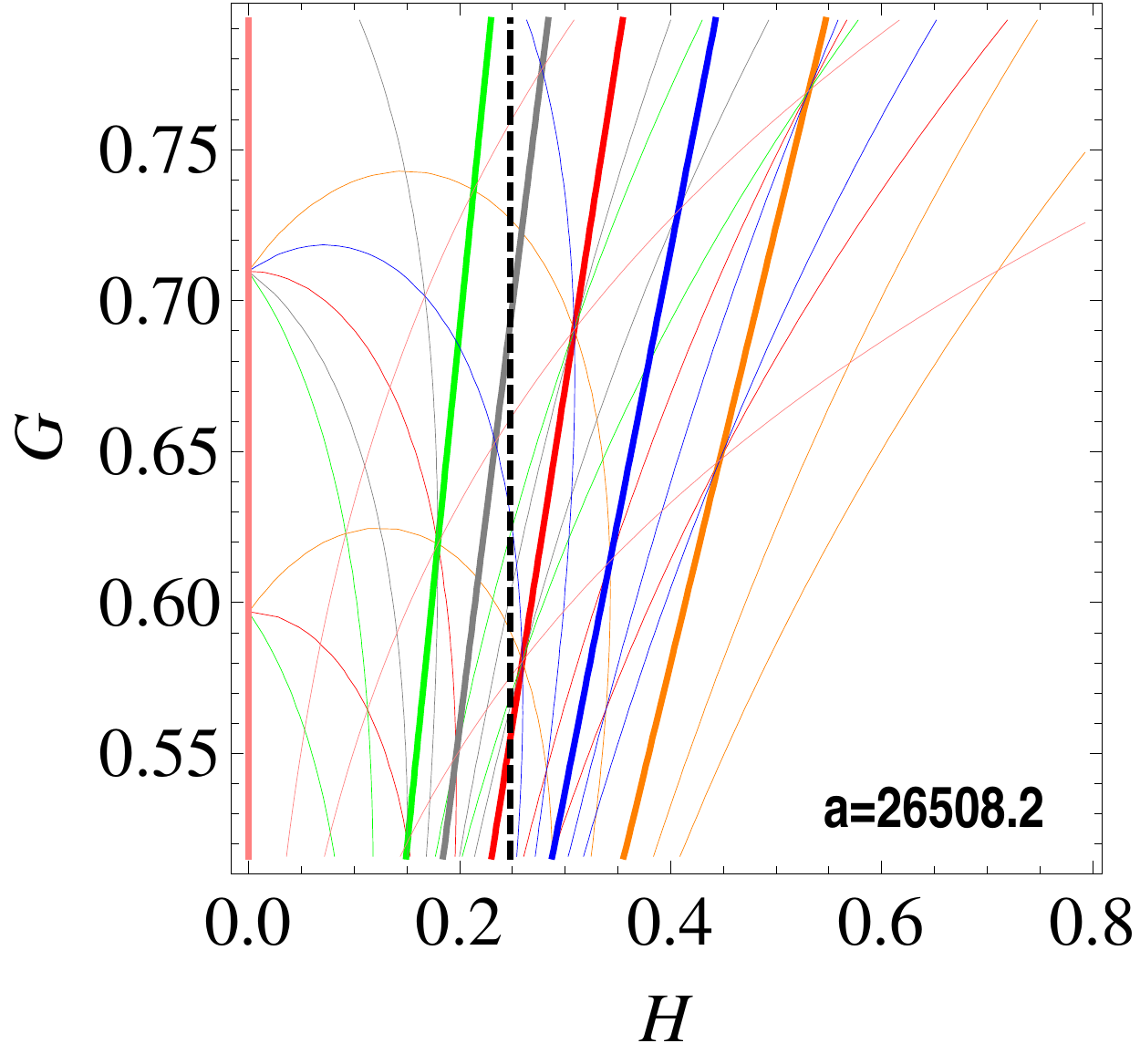}
\vglue0.5cm
\caption{The web structure of resonances in the space of the actions for $a= 13339.1 $ km (upper panels), $a= 18851.7 $ km (bottom left) and  $a= 26508.2 $ km (bottom right). The thick curves represent the location of the following {\it exact resonances} (the multiplet component having $s=0$): $\dot{\Omega}=0$ (pink color, $I=90^\circ$), $\dot{\omega}-\dot{\Omega}=0$ (green color, $I=73.2^\circ$), $2 \dot{\omega}-\dot{\Omega}=0$ (grey color, $I=69.0^\circ$), $\dot{\omega}=0$ (red color, $I=63.4^\circ$), $2\dot{\omega}+\dot{\Omega}=0$ (blue color, $I=56.1^\circ$) and $\dot{\omega}+\dot{\Omega}=0$ (orange color, $I=46.4^\circ$).  The thin curves give the position of the resonances $ (2-2 p) \dot{\omega}+ m \dot{\Omega}+ s \dot{\Omega}_k=0$ with $p,m=0,1,2$ and $s=-2,-1,1,2$.   The vertical black dashed lines correspond to the values of $H$ used in computing the FLI maps. Excluding the top left panel, which is obtained for $G\in [0, G_{max}]$, in the other plots $G$ varies from $G_{min}$ to $G_{max}$, as explained in the text.}
\label{WEB_structure}
\end{figure}

By computing the  Fast Lyapunov Indicators we investigate cartographically the dynamical features of the critical inclination resonance, i.e., the {\it whole resonance} $\dot{\omega}=0$. We recall that the FLI is an
efficient tool to locate the equilibria, to evaluate the width of the resonant islands and to  study the stable and chaotic behavior of a dynamical system by comparing
the values of the FLIs as the initial conditions or parameters are varied.
We present some results in the $(\omega, G)$ plane,
providing the value of the FLI through a color scale, where darker colors denote
regular dynamics, either periodic or quasi–periodic, while lighter colors denote
chaotic motions.

Given $a$ and $H$, we compute a grid of $100 \times 100 $ points of
the $\omega$--$G$ plane, where the argument of perigee ranges in the interval $[0^\circ, 360^\circ]$,
while $G$ spans the interval $[G_{min}, G_{max}]$. However, instead of displaying $G$ on the vertical axis, in each plot we show the eccentricity values (on the left) and the inclination values (on the right), computed by using the relations \eqref{Delaunay} for given values of $a$ and $H$.
In all plots that represent the FLI values, we use the ranges corresponding to those used in the top left and bottom panels of Figure~\ref{WEB_structure} and the top panel of Figure~\ref{FLI_GEO}. The relation among $G$, $e$ and $I$ is trivial; for instance, the values $e=0.46$, $I=63.6^\circ$ from the top panels of Figure~\ref{FLI_Molniya} correspond to the value $G=0.5$ from the top right panel of Figure~\ref{WEB_structure}.

Although the initial conditions are set such that the initial orbits have the perigee larger than $R_E$, since we are interested in finding the equilibrium points and evaluating the area of chaotic regions, during the total time of integration, we neglect the Earth's dimensions.
Namely,
we propagate each orbit up to $465$ years (equal to $25 \times 18.6$ years), even if at some intermediate time the perigee distance becomes smaller than the radius of the Earth.

\vskip.2in

\begin{figure}[h]
\centering
\vglue0.1cm
\hglue0.1cm
\includegraphics[width=7.5truecm,height=6truecm]{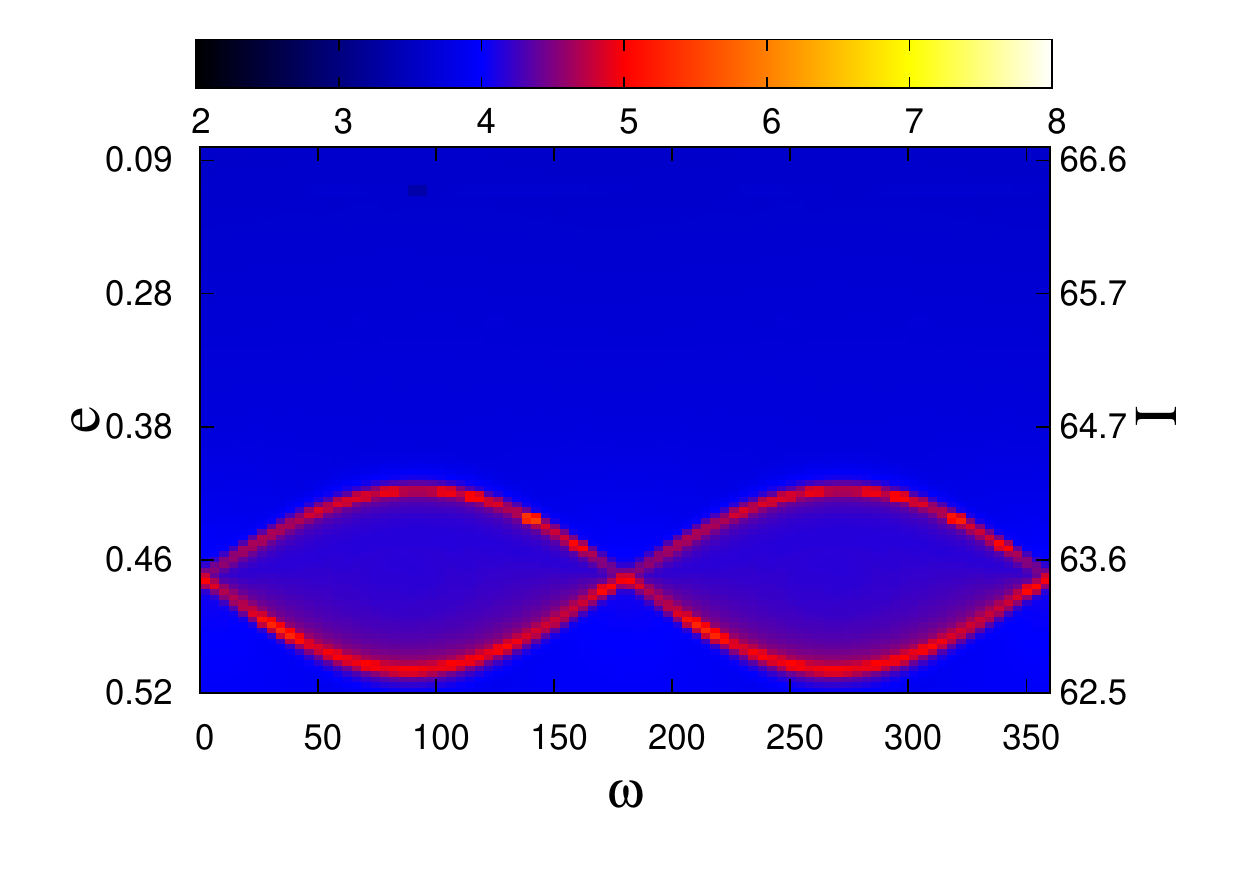}
\includegraphics[width=7.5truecm,height=6truecm]{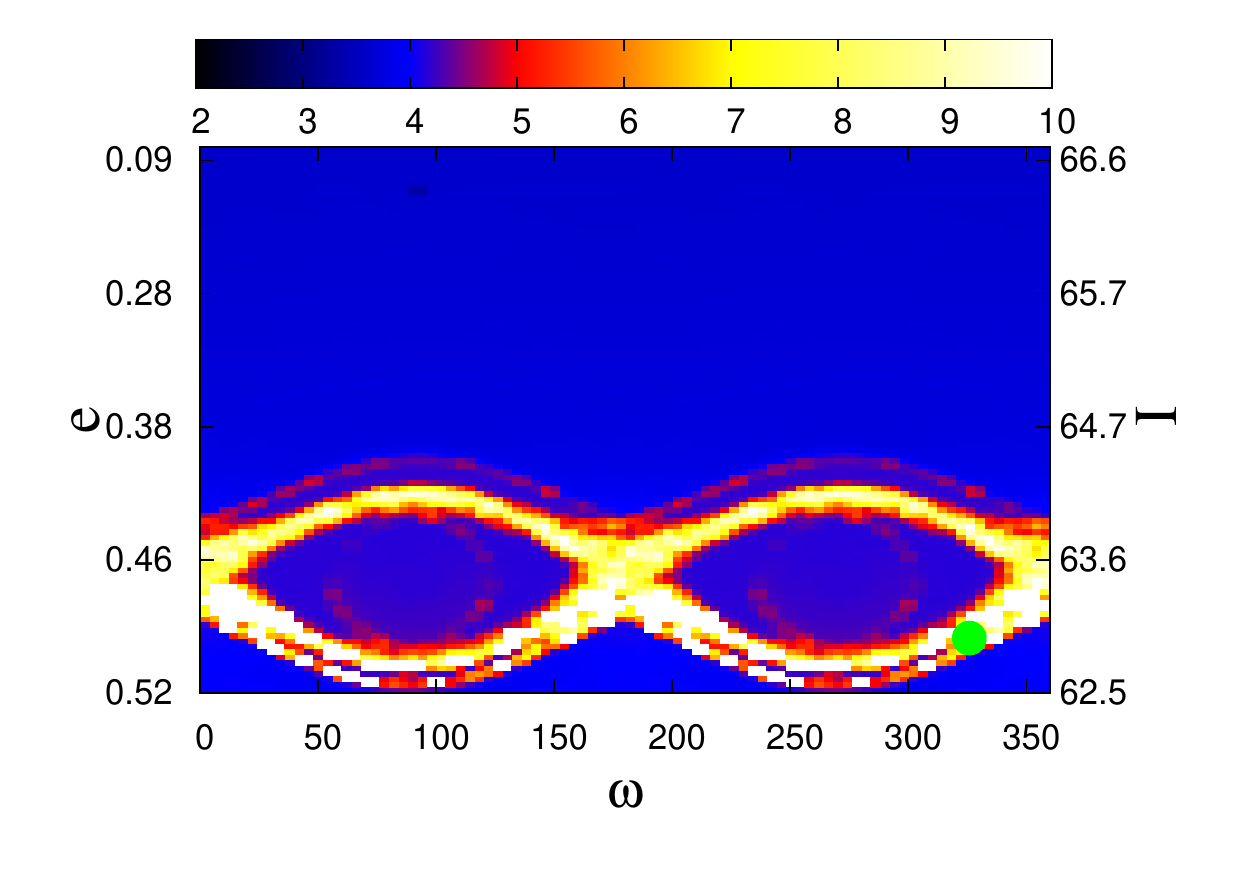}
\includegraphics[width=7.5truecm,height=6truecm]{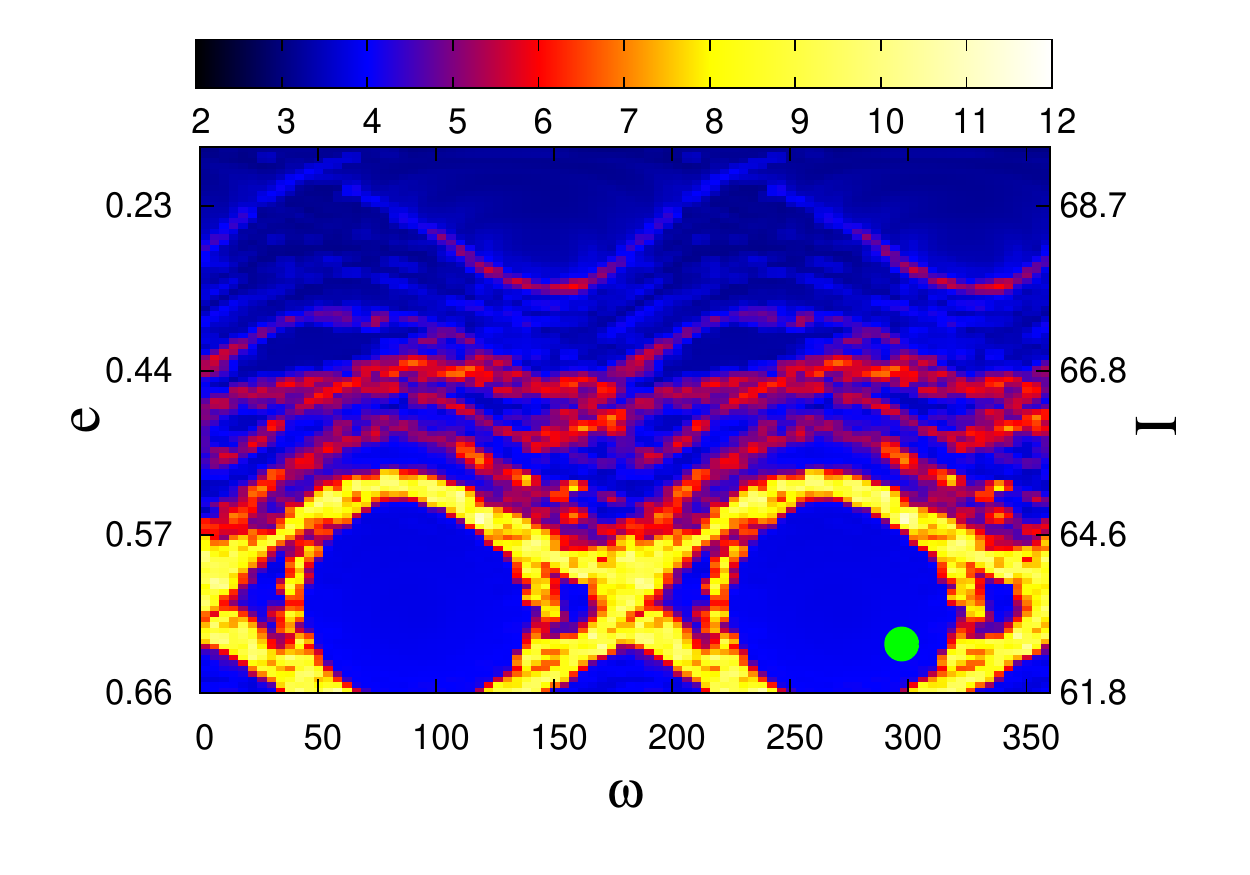}
\includegraphics[width=7.5truecm,height=6truecm]{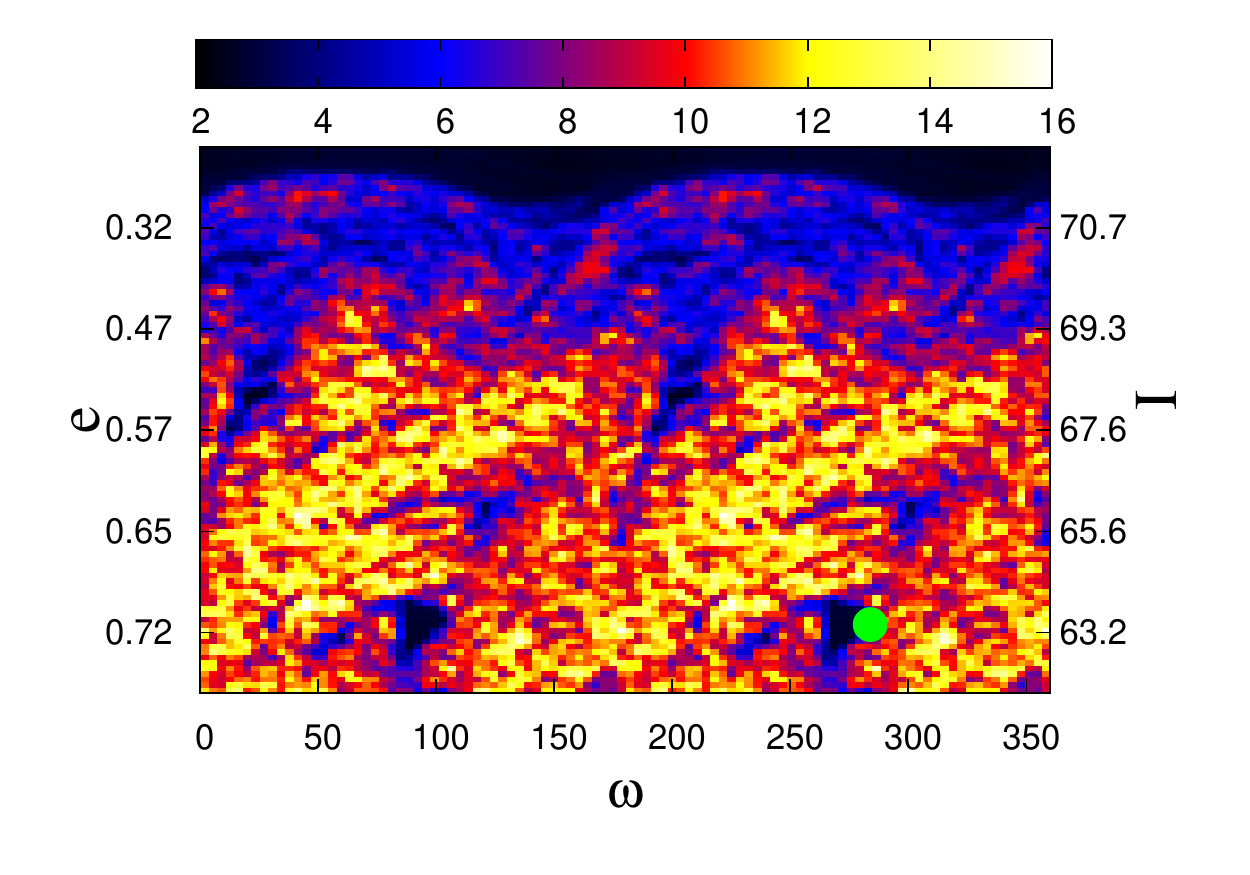}
\vglue0.5cm
\caption{FLIs for the critical inclination secular resonance.
Upper panels: $a= 13339.1 $ km, $\Omega=236.07^\circ$, $H= 0.222$. Top left: $\dot\Omega_{k}=0$. Top right: full model.
Bottom left (full model): $a= 18851.7 $ km, $\Omega=100.66^\circ$, $H= 0.236$. Bottom right (full model): $a= 26508.2 $ km, $\Omega=270.26^\circ$, $H= 0.248$. The green circles represent Molniya 1-86 on the top right panel, Molniya 1-88 in the bottom left panel and Molniya 1-81 on the bottom right panel.}
\label{FLI_Molniya}
\end{figure}

Figure~\ref{FLI_Molniya} shows the FLIs for increasing values of the semi--major axis, i.e. $a=13339.1$ km (top panels), $a=18851.7$ km (bottom left) and $a=26508.2$ km (bottom right). A color scale provides an indication of the behavior of the dynamics: dark colors (i.e., low FLIs)
denote a regular dynamics, while light colors (i.e., high FLIs) correspond to chaotic motions.

The role of the variation of $\Omega_k$ in generating the resonance structure, which is responsible for chaotic motions on timescales of order of tens to hundreds of years, is illustrated graphically in the top panels of Figure~\ref{FLI_Molniya}. These plots are obtained for the following parameters, which have been chosen with the intention to apply our study to Molniya satellites (see Table~\ref{Molniya_elements}): $a= 13339.1 $ km, $\Omega=236.07^\circ$ and $H= 0.222$;
the difference between the two top plots is related to the model used in computing the FLI values. Precisely, in Figure~\ref{FLI_Molniya} top left, $\Omega_k$ is considered constant, while Figure~\ref{FLI_Molniya} top right is obtained for the non--autonomous, two degrees--of--freedom Hamiltonian \eqref{Hamiltonian}. Clearly, the top left panel represents a pendulum--like plot: the stable points are located at $\omega=90^\circ$ and $\omega=270^\circ$, whereas the unstable ones are placed at $\omega=0^\circ$ and $\omega=180^\circ$. Responsible for the existence of the resonant island is, in fact, a single term of the expansion, obtained by combining all resonant terms whose harmonic angles are of the form  $2 \omega\pm s \Omega_k$ with $s=0,1,2$.

When the full model is considered, then the separatrix is filled by chaotic regions (Figure~\ref{FLI_Molniya}, top right panel), revealing thus the interaction between the resonant term having the harmonic angle $2 \omega$ and other terms of the expansion. In fact, as it can be seen from Figure~\ref{WEB_structure} top right, for $H=0.222$ the exact resonance is located at about $G=0.498$ (the intersection between the thick red curve and the vertical dashed black line) and close to it there are just resonances stemming from the same multiplet. A detailed inspection of the expansions of  $\overline{\R}_{Moon}$ and $\overline{\R}_{Sun}$ reveals that the resonant terms, let us call them $\mathfrak{T}_s^{\dot{\omega}=0}$, responsible for the {\it whole resonance} $\dot{\omega}=0$ have the form $\mathfrak{T}_s^{\dot{\omega}=0}=C_s a^2 e^2 (1-\cos^2 I) \cos(2 \omega +s \Omega_k)$, $s=-2,-1,0,1,2$, where $C_s$ depends on the lunar and solar elements. Moreover, it can be seen that $|C_0|$ is at least 20 times greater than $|C_{-1}|$ and $|C_{1}|$, and more than 200 times larger than $|C_{-2}|$ and $|C_{2}|$. As a result, the width of the resonant island associated to $\mathfrak{T}_0^{\dot{\omega}=0}$ is much larger than the width of the other islands. As it can be seen in Figure~\ref{WEB_structure} top right panel, the dashed black line intersects (in order from bottom to top) three exact resonances:
$2\dot{\omega}+\dot{\Omega}_k=0$ at approximately $G=0.485$,
$\dot{\omega}=0$ at $G=0.495$ and $2\dot{\omega}-\dot{\Omega}_k=0$ at about $G=0.52$. The structure shown in Figure~\ref{FLI_Molniya} top right panel is the result of the overlapping of these three exact resonances.

\begin{table}[h]
\begin{tabular}{|l|}
  \hline
\footnotesize{\texttt{MOLNIYA 1-81}}\\
\footnotesize{\texttt{1\ 21426U\ 91043A\ \ \ 15256.55204240\ -.00000042\ \ 00000-0\ -18490-1\ 0\ \ 9994}}\\
\footnotesize{\texttt{2\ 21426\ \ 63.3807\ 270.2557\ 7154024\ 283.9028\ 344.3128\ \ 2.00606557177626}}\\
\hline
\footnotesize{\texttt{MOLNIYA 1-88}}\\
\footnotesize{\texttt{1\ 23420U\ 94081A\ \ \ 15255.51665308\ \ .00000061\ \ 00000-0\ \ 41968-4\ 0\ \ 9994}}\\
\footnotesize{\texttt{2\ 23420\ \ 62.8537\ 100.6611\ 6341703\ 297.1923\ \ 12.8801\ \ 3.34494686164329}}\\
\hline
\footnotesize{\texttt{MOLNIYA 1-86}}\\
\footnotesize{\texttt{1\ 22671U\ 93035A\ \ \ 15256.86281272\ \ .00000717\ \ 00000-0\ \ 28094-3\ 0\ \ 9994}}\\
\footnotesize{\texttt{2\ 22671\ \ 62.9189\ 236.0661\ 4962239\ 325.8722\ 222.6630\ \ 5.61987431213385}}\\
 \hline
  \end{tabular}
 \vskip.2in
 \caption{Two Line Element set (\citet{citeTLE}) for the following satellites:  Molniya 1-81, Molniya 1-88, Molniya 1-86, during the middle of September  2015 (see columns 19-32 of the first line of the each set for the exact Epoch). Data are taken from \cite{citeTLE}.}\label{TLE}
\end{table}

\begin{table}[h]
\begin{tabular}{|c|c|c|c|c|c|c|c|c|}
\hline
Satellite & a (in km) & $e$ & $I$ & $\Omega$ & $\omega $ & $L$ & $G$ & $ H $\\
   \hline
   \hline
   Molniya 1-81 & $26508.2$ & $0.7154$ & $63.38$ & $270.26$ & $283.90 $ & $0.793$ & $0.554$ & $ 0.248 $\\
   \hline
   Molniya 1-88 & 18851.7 & $0.6342$ & $62.85$ & $100.66$ & $297.19 $ & $0.669$ & $0.517$ & $ 0.236 $\\
   \hline
   Molniya 1-86 & 13339.1 & $0.4962$ & $62.92$ & $236.07$ & $325.87 $ & $0.562$ & $0.488$ & $ 0.222 $\\
   \hline
   \end{tabular}
 \vskip.2in
 \caption{The orbital elements and the corresponding action variables for the following satellites:  Molniya 1-81, Molniya 1-88, Molniya 1-86. For each satellite, the values are determined at the corresponding Epoch specified in Table~\ref{TLE}. The angles $I$, $\Omega$ and $\omega$ are expressed in degrees, while the units of length and time used to compute the actions $L$, $G$, $H$ are described in the text.
Data are taken from \cite{citeTLE}.}\label{Molniya_elements}
\end{table}

Increasing the value of the semi--major axis, the resonant curves in Figure~\ref{WEB_structure} bottom panels become increasingly scattered, their intersections suggesting a more intricate dynamics. In the bottom plots of Figure~\ref{WEB_structure}, the vertical dashed black line crosses more resonant curves than in the top right panel of Figure~\ref{WEB_structure}. In particular it also intersects the curves associated to the {\it exact resonances} $2\dot{\omega} -\dot{\Omega}=0$ and $2\dot{\omega} -\dot{\Omega}+\dot{\Omega}_k=0$. An estimation of the magnitude of the resonant terms shows that the width of these two {\it exact resonances} is much larger than the width of any other resonance intersected by the dashed vertical line, excluding of course the exact resonance $\dot{\omega}=0$. In the bottom right panel of Figure~\ref{FLI_Molniya}, the {\it exact resonances} $\dot{\omega}=0$, $2\dot{\omega} -\dot{\Omega}=0$ and $2\dot{\omega} -\dot{\Omega}+\dot{\Omega}_k=0$ overlap, whereas in Figure~\ref{FLI_Molniya} bottom left they do not interact. The regular and chaotic regions in the bottom panels of Figure~\ref{FLI_Molniya} are the consequence of the complex structure of these resonances.

\begin{figure}[h]
\centering
\vglue0.1cm
\hglue0.1cm
\includegraphics[width=5.5truecm,height=5truecm]{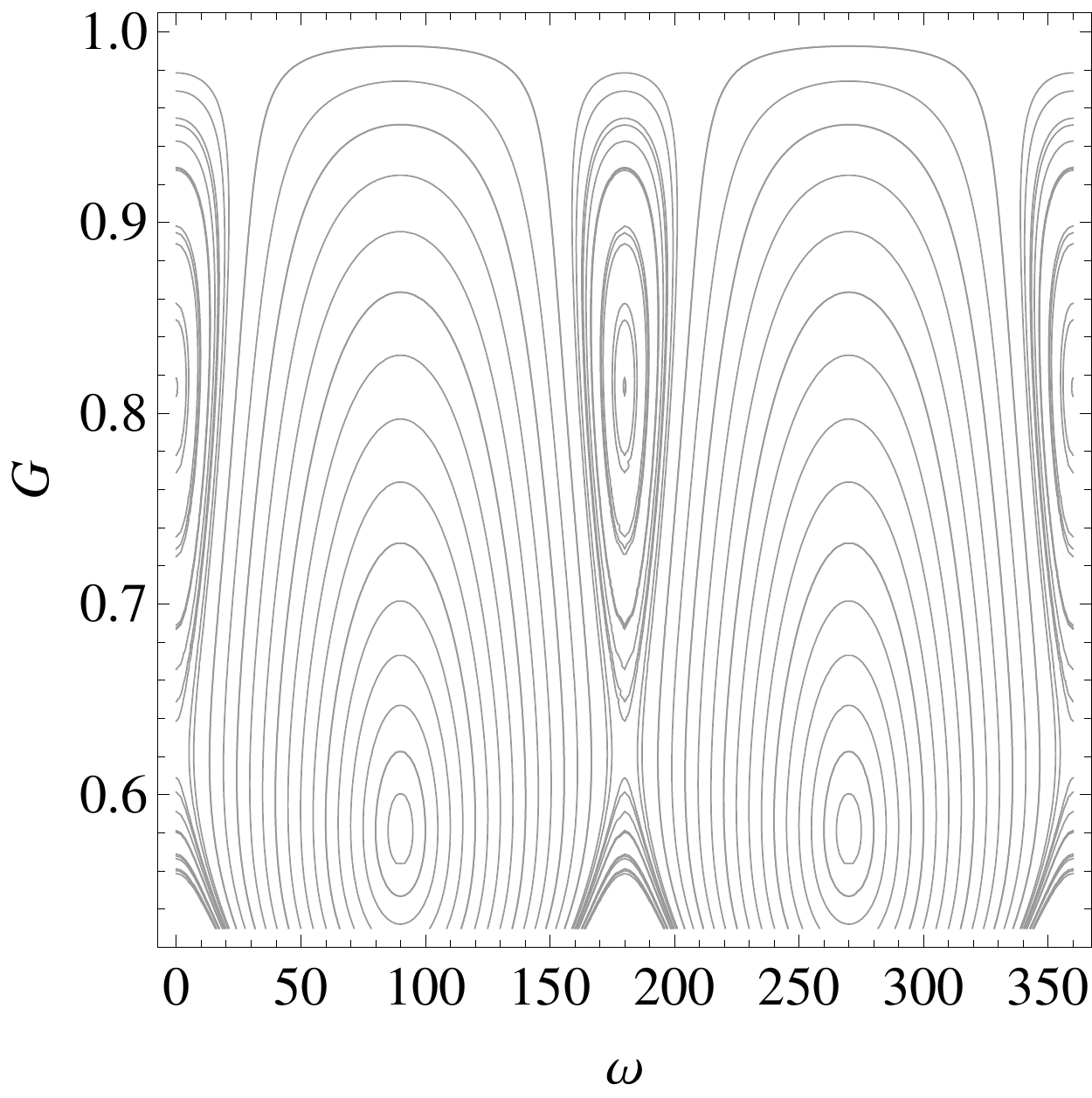}\\
\includegraphics[width=7.5truecm,height=6truecm]{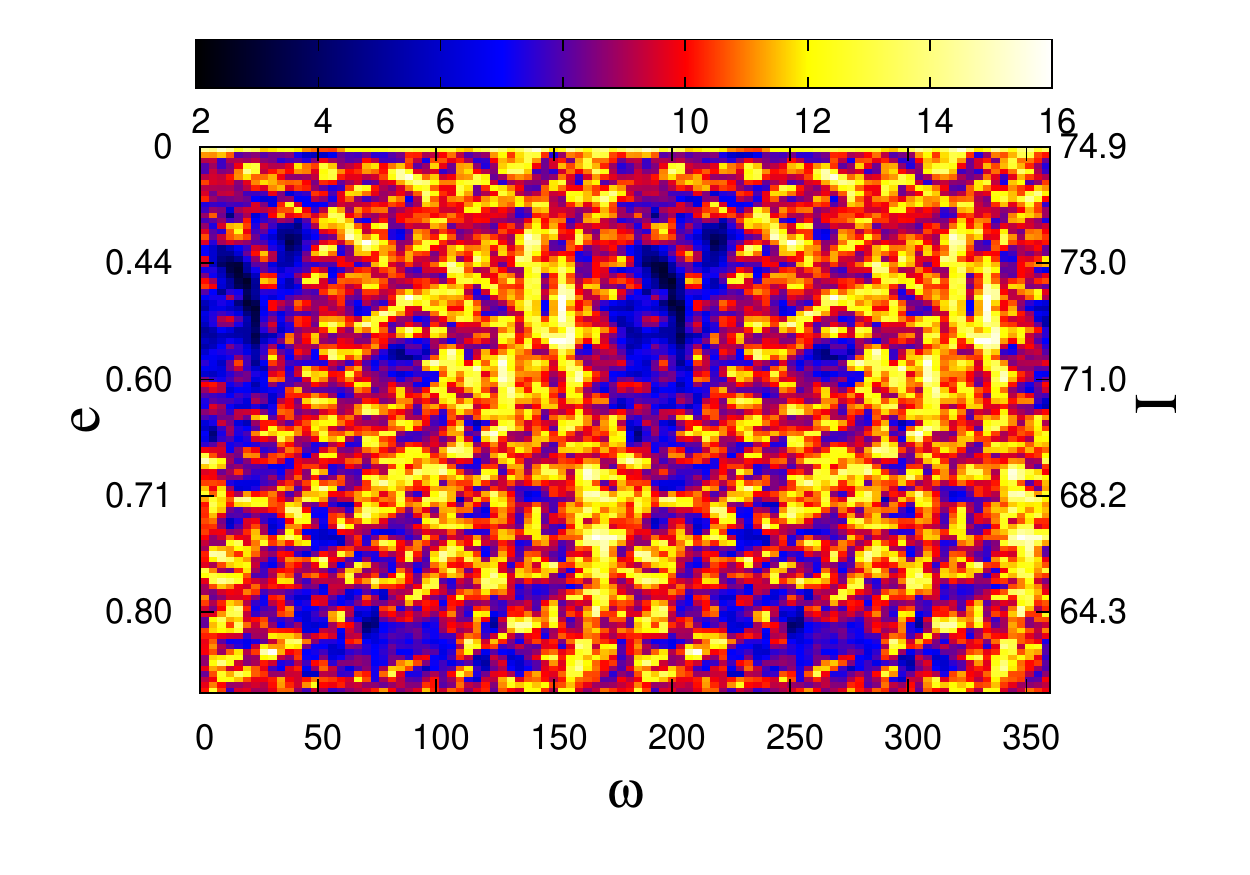}
\includegraphics[width=7.5truecm,height=6truecm]{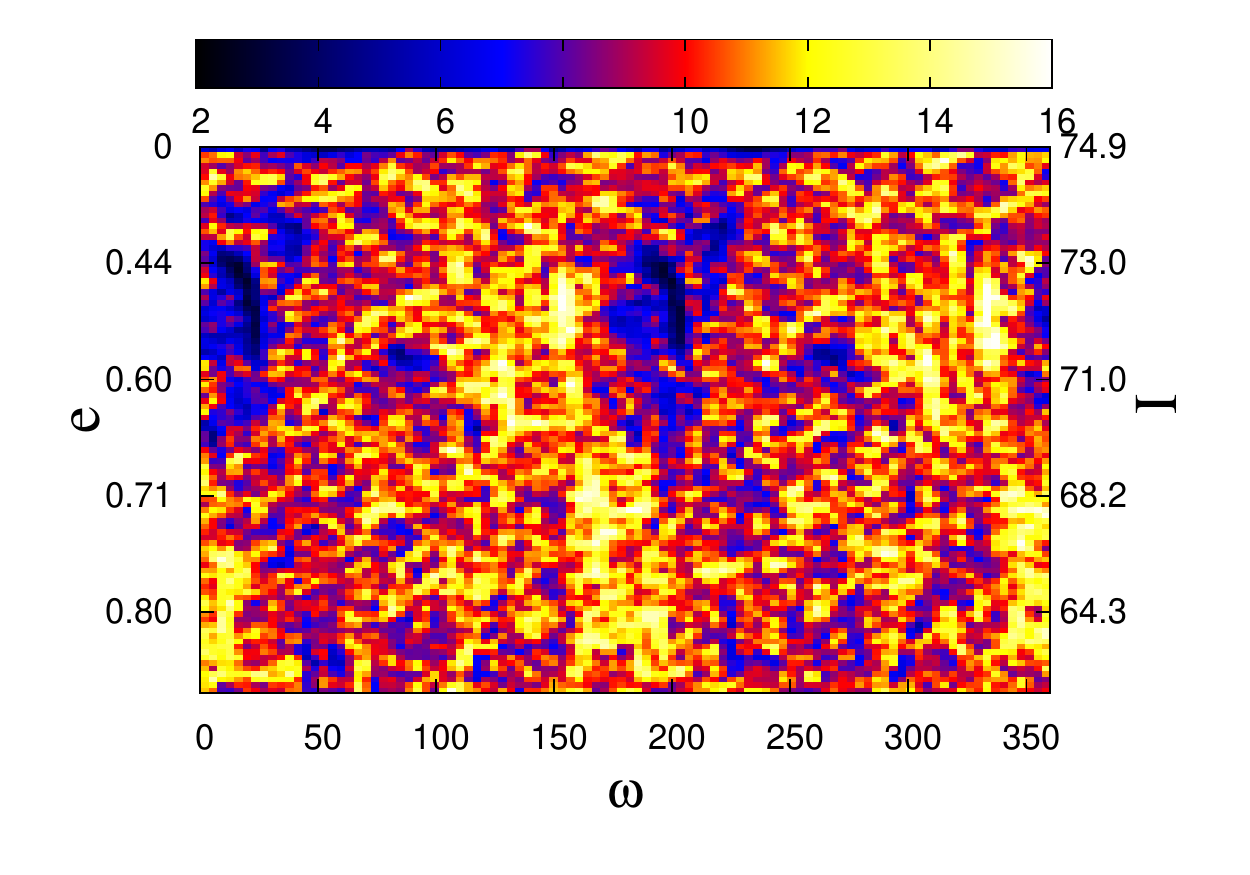}
\vglue0.5cm
\caption{The critical inclination resonance for $a=42164.17$ km and $H=0.26$.  Top: Phase portrait  for the integrable system obtained by averaging $\H$ over $\Omega$ and considering $\Omega_{k}$ as a constant.
Bottom:  FLI values for the full model \eqref{Hamiltonian}. In the bottom left
panel the series expansion in $a/a_k$ from \equ{eq:lane} is performed to the quadrupole
($l=2$) order, while in the bottom right plot, this series is taken to the octupole
($l = 3$) order.} The initial value of $\Omega$ is zero for the bottom panels.
\label{FLI_GEO}
\end{figure}

The green circles in Figure~\ref{FLI_Molniya} represent the satellites: Molniya 1-81 (bottom right), Molniya 1-88 (bottom left) and Molniya 1-86 (top right). Table~\ref{TLE} contains the Two Line Element set for these satellites, while Table~\ref{Molniya_elements} gives their corresponding orbital elements and action variables. From \citet{citeTLE}, during the middle of September 2015, the mean motion of almost all Molniya satellites is around 2 revolutions per day (or equivalently the semi--major axis is about $26\,000$ km), excluding the satellites Molniya 1-88, Molniya 1-86, Molniya 1-S (placed in a geosynchronous orbit whose inclination is $I=10^\circ$) and Molniya 1-44 (located at the edge of the LEO region). Moreover, the argument of perigee of these satellites is close to $270^\circ$; therefore, they are contained in a stable region that guarantees small excursions in eccentricity and $\omega$ (Figure~\ref{FLI_Molniya} bottom right). Molniya 1-88 is also in a regular region (Figure~\ref{FLI_Molniya} bottom left), and although Molniya 1-86 is in a chaotic zone, it cannot collide with the Earth since the resonant island is located above the line $G=G_{min}$ (Figure~\ref{FLI_Molniya} top right).

For larger semi--major axes, let us say for the GEO region, the problem increases in complexity. First of all, we can no longer approximate $\dot{\omega}$, $\dot{\Omega}$  by the simple relations \eqref{omega12} because the magnitude of the perturbing forces due to Moon and Sun becomes comparable with that due to $J_2$. Thus, in obtaining the network of resonances, one should include also the effects induced by Moon and Sun. Moreover, a bifurcation phenomenon enters into the scene, showing that there are some cases when a specific resonance cannot be modeled by a pendulum--type system, but one should use a so--called \sl extended fundamental model, \rm as described in \citet{sB01}.

Indeed, let us consider the integrable Hamiltonian system, obtained by averaging \eqref{Hamiltonian} over $\Omega$ and by considering $\Omega_k$ as a constant. Roughly speaking, we have isolated the resonance $\dot{\omega}=0$ from any other resonance. For such a reduced one degree--of--freedom Hamiltonian, the phase portrait is given in Figure~\ref{FLI_GEO} top panel. This plot, obtained for $a=42\,164.17$ km and $H=0.26$, shows the existence of some additional elliptic equilibrium points at about $G=0.82$ for $\omega=0^\circ, 180^\circ$. We notice that these equilibrium points are not revealed in the previous studies concerning the critical inclination resonance in the GEO region \citep[see][]{fDaM93}. The reason is related to the use of different values for the energy level. For instance, if we draw the phase portrait for, let us say $H=0.3$, then we obtain just the equilibrium points of the primary resonance.

Figure~\ref{FLI_GEO} bottom left is obtained for the Hamiltonian \eqref{Hamiltonian} (the lunar and solar expansions are taken up to $l=2$). Although the overlapping of resonances leads to chaotic motions, we notice that some stable regions remain  around each equilibrium stable point.

So far, we have restricted our analysis to the quadrupolar gravitational interactions, namely we expanded the lunisolar disturbing function up to the second order ($l=2$) in the ratio of semi--major axes. For the Sun, the effect of the third harmonic (parallactic, or octupole) term is negligible, while for the Moon the contribution of the third harmonic term becomes more important as the orbital radius increases.
Neglecting this term sets an upper limit to the radius of the satellite orbit for which the preceding theory is valid \citep{mL63}. In considering the long--range stability at GEO (where $a/a_k \sim 0.1$), \citet{pM61} advises to take the parallactic term into consideration, and, in fact, \citet{yLsN11} found that interesting dynamical behaviors can occur for the Kozai--Lidov cycles in exoplanetary studies under octupole--order secular interactions. While in the case including the octupole, the form of the perturbations cannot be easily expressed analytically, because the terms depending upon the position of the Moon's perigee will appear, we can provide a numerical analysis.

Figure~\ref{FLI_GEO} bottom right shows the FLI values at the octupole level of approximation, which, when compared to the bottom left panel of Figure~\ref{FLI_GEO}, lead us to conclude that the effects of $(a/a_k)^3$ terms, and the corresponding higher--order secular resonances stemming from the lunar perigee,
are indeed important at GEO, at least for large eccentricities.

\section{Conclusions}\label{sec:conclusions}

As it often happens in Celestial Mechanics, analytical and numerical methods are complementary, and should both be used to get a detailed description of the dynamics. As far as artificial satellite or space debris dynamics is concerned, the main gravitational effects are due to the geopotential, the influence of the Moon and the attraction of the Sun. It is therefore mandatory to have a correct expansion of the lunisolar potential in terms of the appropriate elements. Such expansion is provided in Proposition~\ref{pro:Lane}, whose formula can be implemented on an algebraic manipulator to compute a series expansion to a finite order in the parameters of the model (typically, the eccentricity and the ratio of the semi--major axes).

The advantage of having an explicit analytical expansion of the lunisolar potential is clarified by a concrete application concerning lunisolar secular resonances for the Molniya satellites, which move on orbits with high eccentricity and inclination. As shown in Section~\ref{sec:averages}, the analytical expansion gives a powerful tool to investigate the dynamics, especially when looking at the interaction between resonant motions. In fact, the analytical expansion provides a straightforward way to compute the resonance relations and, in particular, the interaction between different resonances.

The explicit series development also allows us to highlight the role of each component of the multiplet associated to a specific resonance. The analytical approach yields information about the dynamics and, precisely, on the chaotic behavior generated by the overlapping of resonances. The numerical results provided by the computation of the FLIs gives a global view of the dynamics, whose fine structure was obtained by analyzing the lunisolar expansion. We remark that this is just one of the many applications which can be obtained as a byproduct of the analytical expansion of the lunisolar potential and it is given as a possible motivation of the present work.

\section*{Appendix}

\begin{appendices}


\section{Hansen's coefficients}
\label{app:Hansen}

The Hansen coefficients $H_{lpq}(e)=X_{l-2p+q}^{l, l-2p}(e)$, $G_{lhj}(e^{*})=X_{l-2h+j}^{-(l+1), l-2h}(e^*)$ may be computed easily with an algebraic manipulator through the following formula \citep[see][]{mJ65, gG76}:

\beqano
X_{k}^{n,m}(e)&=&\biggl(1+\Bigl({e\over {1+\sqrt{1-e^2}}}\Bigr)^2\biggr)^{-n-1}\ \sum_{s=0}^{s_1} \ \sum_{t=0}^{t_1} \ \left(\begin{array}{c}
  n-m+1 \\
  s \\
 \end{array}\right)\
\left(\begin{array}{c}
  n+m+1 \\
  t \\
 \end{array}\right) \nonumber\\
&&\times \Bigl(-{e\over {1+\sqrt{1-e^2}}}\Bigr)^{s+t}\ J_{k-m-s+t}( k e)\ ,
\eeqano
where $n$, $m$, $k$ are integers, $J_b$ denotes the Bessel function of the first kind and $s_1$, $t_1$ are defined by
\beqano
s_1=\left\{
\begin{array}{ll}
  n-m+1, & \quad \textrm{if\ }  n-m+1 \geq 0  \\
  \infty, &\quad \textrm{if\ }  n-m+1 < 0 \\
 \end{array} \right.\ , \qquad
 t_1=\left\{
\begin{array}{ll}
  n+m+1, & \quad \textrm{if\ }  n+m+1 \geq 0  \\
  \infty, &\quad \textrm{if\ }  n+m+1 < 0 \\
 \end{array} \right.\ .
\eeqano

\section{Kaula's inclination functions}
\label{app:inclination_functions}

The Kaula's inclination function, denoted by $F_{lmp}(I)$, has the form \citep[see][]{wK62, wK66}:
\beqano
F_{lmp}(I)&=&\sum_{t=0}^{\min{\{p,[{{l-m}\over 2}]}\}} {{(2l-2t)!}\over {t!(l-t)!(l-m-2t)!2^{2l-2t}}} \sin^{l-m-2t}I\ \sum_{s=0}^m\left(\begin{array}{c}
  m \\
  s \\
 \end{array}\right)
 \cos^sI\nonumber\\
 &&\times \sum_c \left(\begin{array}{c}
  l-m-2t+s \\
  c \\
 \end{array}\right)
\left(\begin{array}{c}
  m-s \\
  p-t-c \\
 \end{array}\right)
 (-1)^{c-[{{l-m}\over 2}]}\ ,
\eeqano
where $[\cdot]$ denotes the integer part and $c$ is summed over all values for which the binomial coefficients are not zero.

\section{Proof of Lemma~\ref{lem:CS}}\label{app:lemCS}

We present the proof of Lemma~\ref{lem:CS}.

\vskip.1in

Let us first notice that the following relation holds:
\begin{align}
    \label{eq:Legendre_prop}
    P_l^{-s} (x) = (-1)^s \frac{(l - s)!}{(l + s)!} P_l^s (x)
\end{align}
and that $A_l^{m,s}=A_l^{m,-s}$.
From the definitions $C_l^m \equiv A_l^m \cos m \alpha^\prime$,
$S_l^m \equiv A_l^m \sin m \alpha^\prime$ and from Lemma~\ref{lem:plm}, one obtains:
\begin{align}
    C_l^m + i S_l^m
    \nonumber
    & = A_l^m \cos m \alpha^\prime + i A_l^m \sin m \alpha^\prime = A_l^m e^{i m \alpha^\prime} \\
    \nonumber
    & = \frac{\mathcal{G} m_k \epsilon_m (l - m)!}{a_k (l + m)!} \left( \frac{a}{a_k} \right)^l
        \left( \frac{r}{a} \right)^l \left( \frac{a_k}{r_k} \right)^{l+1} P_l^m
        \left( \sin \delta^\prime \right) e^{i m \alpha^\prime} \\
    \nonumber
    & = \frac{\mathcal{G} m_k \epsilon_m (l - m)!}{a_k (l + m)!} \left( \frac{a}{a_k} \right)^l
        \left( \frac{r}{a} \right)^l \left( \frac{a_k}{r_k} \right)^{l+1}
        \sum\limits_{s=-l}^l \Lambda_l^{m, s} P_l^s (\sin \delta_k) e^{i s \alpha_k} \\
    \nonumber
    & = \frac{\mathcal{G} m_k \epsilon_m}{a_k (l + m)!} \left( \frac{a}{a_k} \right)^l
        \left( \frac{r}{a} \right)^l \left( \frac{a_k}{r_k} \right)^{l+1}
        \sum\limits_{s=-l}^l P_l^s (\sin \delta_k) (l - s)! U_l^{m,s} e^{i m \pi/2}
        e^{i s (\alpha_k - \pi/2)} \\
    \label{eq:C_imag_S_tmp}
    & = (i)^m \sum\limits_{s=-l}^l (l - s)! \frac{A_l^{m,s}}{\epsilon_s} P_l^s (\sin \delta_k)
        U_l^{m,s} e^{i s (\alpha_k - \pi/2)},
\end{align}
where we used $e^{i m \pi /2} = (i)^m$. Setting $\beta \equiv
\alpha_k - \pi/2$, from \equ{eq:CSA} one has:
\begin{align*}
    C_l^{m,s} \cos (s \beta) + i S_l^{m,s} \sin (s \beta)
    & = \frac{1}{2} \left( U_l^{m,s} + (-1)^s U_l^{m,-s} \right) \frac{e^{i s \beta} + e^{-i s \beta}}{2}\\
    & + \frac{i}{2} \left( U_l^{m,s} - (-1)^s U_l^{m,-s} \right) \frac{e^{i s \beta} - e^{-i s \beta}}{2 i} \\
    & = \frac{1}{2} \left( U_l^{m,s} e^{i s \beta} + (-1)^s U_l^{m,-s} e^{-i s \beta} \right),
\end{align*}
and, due to \eqref{eq:Legendre_prop}, we obtain
\eqref{eq:C_imag_S}; hence, we can split \eqref{eq:C_imag_S} as
\begin{align*}
    C_l^m + i S_l^m
    & = i^m \sum\limits_{s = 0}^l (l - s)! \frac{A_l^{m,s}}{2} P_l^s (\sin \delta_k)
        \left( U_l^{m,s} e^{i s \beta} + (-1)^s U_l^{m,-s} e^{-i s \beta} \right) \\
    & = i^m \sum\limits_{s = -l}^l \frac{A_l^{m,s}}{\epsilon_s} (l - s)! P_l^s (\sin \delta_k)
        U_l^{m,s} e^{i s \beta}\ ,
\end{align*}
which coincides with \eqref{eq:C_imag_S_tmp} above.

\vskip.1in


\begin{remark}
Due to the fact that
\begin{align*}
    i^m = \left\{ \begin{array}{cl} (-1)^{m/2} & \text{$m$ even}\ ,
        \\[0.3em] i (-1)^{(m-1)/2} & \text{$m$ odd}\ , \end{array} \right.
\end{align*}
 we obtain:
 \begin{align}
    (-1)^{k_1} \sum\limits_{s = 0}^l (l - s)! A_l^{m,s} P_l^s
    (\sin \delta_k) C_l^{m,s} \cos (s (\alpha_k - \pi/2))
    & = \left\{ \begin{array}{cl} C_l^m & \text{$m$ even}\ , \\[0.3em] S_l^m & \text{$m$ odd}\ , \end{array}\right. \nonumber\\
    (-1)^{k_1} \sum\limits_{s = 0}^l (l - s)! A_l^{m,s} P_l^s
    (\sin \delta_k) S_l^{m,s} \sin (s (\alpha_k - \pi/2))
    & = \left\{ \begin{array}{cl} S_l^m & \text{$m$ even}\ , \\[0.3em] -C_l^m & \text{$m$ odd} \end{array} \right.\nonumber
 \end{align}
with $k_1 = [{m/2}]$.
\end{remark}

\section{Proof of Lemma~\ref{lem:theta}}\label{app:Lemma7}

We present the proof of Lemma~\ref{lem:theta}.

\vskip.1in

From \citet[p. 34, equation (3.61)]{wK66}, we have that
\begin{align}
    P_l^s (\sin \delta_k) \cos (s (\alpha_k - \pi/2))
    = (-1)^s \sum\limits_{q = o}^l F_{lsq} (I_k)\times\ \left\{ \begin{array}{cl}
        \cos \theta_{lsq}^\prime & l - s\ even\ , \nonumber\\[0.3em]
        \sin \theta_{lsq}^\prime & l - s\ odd\ . \end{array} \right.\nonumber
\end{align}
Therefore, we obtain:
\begin{align}
(-1)^{k_1} \sum\limits_{s = 0}^l (-1)^s\ (l - s)! A_l^{m,s} \sum\limits_{q = 0}^l F_{lsq} (I_k)\times\
    \left\{ \begin{array}{cl} C_l^{m,s} \cos \theta_{lsq}^\prime & \text{$l - s$ even} \nonumber\\[0.3em]
    C_l^{m,s} \sin \theta_{lsq}^\prime & l - s\ odd \end{array} \right. =
    \left\{ \begin{array}{cl} C_l^m & m\ even\ , \\[0.3em] S_l^m & m\ odd\ , \end{array} \right.\nonumber
\end{align}
\begin{align}
   (-1)^{k_1} \sum\limits_{s = 0}^l (-1)^s\  (l - s)! A_l^{m,s} \sum\limits_{q = 0}^l F_{lsq} (I_k)\times\
    \left\{ \begin{array}{cl} S_l^{m,s} \sin \theta_{lsq}^\prime & \text{$l - s$ even} \nonumber\\[0.3em]
    -S_l^{m,s} \cos \theta_{lsq}^\prime & l - s\ odd \end{array} \right. =
    \left\{ \begin{array}{cl} S_l^m & m\ even\ , \\[0.3em] -C_l^m & m\ odd\ . \end{array} \right.\nonumber
\end{align}

We proceed to prove \equ{eq:Theta_even}.

$ $ \newline Case $m$ even, $l - m$ even, $l - s$ even. From
\eqref{eq:kaula} we have
\begin{align*}
    \mathcal{R}_k
    & = \sum\limits_{l \geq 2} \sum\limits_{m = 0}^l (-1)^m\ \sum\limits_{p = 0}^l F_{lmp} (I)
    \left( C_l^m \cos \theta_{lmp} + S_l^m \sin \theta_{lmp} \right) \\
    & = \sum\limits_{l \geq 2} \sum\limits_{m = 0}^l \sum\limits_{p = 0}^l \sum\limits_{s = 0}^l
        \sum\limits_{q = 0}^l (-1)^{m+s}\ \Big[ (-1)^{k_1} A_l^{m, s} (l - s)! F_{lmp} (I) F_{lsq} (I_k)\\
        & \hspace{12pt} \times \left( C_l^{m, s} \cos \theta_{lmp} \cos \theta_{lsq}^\prime
        + S_l^{m, s} \sin \theta_{lmp} \sin \theta_{lsq}^\prime \right) \Big] \\
    & = \sum\limits_{l \geq 2} \sum\limits_{m = 0}^l \sum\limits_{p = 0}^l \sum\limits_{s = 0}^l
        \sum\limits_{q = 0}^l (-1)^{m+s}\ \left[ (-1)^{k_1} A_l^{m, s} (l - s)! F_{lmp} (I) F_{lsq} (I_k) \right. \\
    & \hspace{12pt} \left. \times \frac{1}{2} \left( U_l^{m, s} \cos \left( \theta_{lmp} - \theta_{lsq}^\prime \right)
        + (-1)^s U_l^{m, -s} \cos \left( \theta_{lmp} + \theta_{lsq}^\prime \right) \right) \right]\ .
\end{align*}
Case $m$ even, $l - m$ even, $l - s$ odd:
\begin{align*}
    \mathcal{R}_k
    & = \sum\limits_{l \geq 2} \sum\limits_{m = 0}^l \sum\limits_{p = 0}^l (-1)^m\ F_{lmp} (I)
        \left( C_l^m \cos \theta_{lmp} + S_l^m \sin \theta_{lmp} \right) \\
    & = \sum\limits_{l \geq 2} \sum\limits_{m = 0}^l \sum\limits_{p = 0}^l \sum\limits_{s = 0}^l
        \sum\limits_{q = 0}^l (-1)^{m+s}\ \Big[ (-1)^{k_1} A_l^{m, s} (l - s)! F_{lmp} (I) F_{lsq} (I_k) \\
    & \hspace{12pt} \times \Big( C_l^{m, s} \sin \theta_{lsq}^\prime \cos \theta_{lmp}
        - S_l^{m, s} \cos \theta_{lsq}^\prime \sin \theta_{lmp} \Big) \Big] \\
    & = \sum\limits_{l \geq 2} \sum\limits_{m = 0}^l \sum\limits_{p = 0}^l \sum\limits_{s = 0}^l
        \sum\limits_{q = 0}^l (-1)^{m+s}\ \left[ (-1)^{k_1} A_l^{m, s} (l - s)! F_{lmp} (I) F_{lsq} (I_k) \right. \\
    & \hspace{12pt} \left. \times \frac{1}{2} \left( -U_l^{m, s} \sin \left( \theta_{lmp} - \theta_{lsq}^\prime \right)
        + (-1)^s U_l^{m, -s} \sin \left( \theta_{lmp} + \theta_{lsq}^\prime \right) \right) \right]\ .
\end{align*}

The proof of \eqref{eq:Theta_even_odd} is given as follows.

$ $ \newline Case $m$ even, $l - m$ odd:
\begin{align*}
    \mathcal{R}_k
    \nonumber
    & = \sum\limits_{l \geq 2} \sum\limits_{m = 0}^l \sum\limits_{p = 0}^l (-1)^m\ F_{lmp} (I)
        \left( -S_l^m \cos \theta_{lmp} + C_l^m \sin \theta_{lmp} \right) \\
    & = \sum\limits_{l \geq 2} \sum\limits_{m = 0}^l \sum\limits_{p = 0}^l \sum\limits_{s = 0}^l
        \sum\limits_{q = 0}^l (-1)^{m+s}\ (-1)^{k_1} A_l^{m, s} (l - s)! F_{lmp} (I) F_{lsq} (I_k) \\
    & \hspace{12pt} \times
        \left\{ \begin{array}{cl}
            -S_l^{m, s} \cos \theta_{lmp} \sin \theta_{lsq}^\prime
            + C_l^{m, s} \sin \theta_{lmp} \cos \theta_{lsq}^\prime &l - s\ even \\[0.5em]
            S_l^{m, s} \cos \theta_{lmp} \cos \theta_{lsq}^\prime
            + C_l^{m, s} \sin \theta_{lmp} \sin \theta_{lsq}^\prime & l - s\ odd
        \end{array} \right. \\
    & =\sum\limits_{l \geq 2} \sum\limits_{m = 0}^l \sum\limits_{p = 0}^l \sum\limits_{s = 0}^l
        \sum\limits_{q = 0}^l (-1)^{m+s}\  (-1)^{k_1} A_l^{m, s} (l - s)! F_{lmp} (I) F_{lsq} (I_k) \\
    & \hspace{12pt} \times
        \left\{ \begin{array}{cl}
            \frac{1}{2} \left[ U_l^{m, s} \sin \left( \theta_{lmp} - \theta_{lsq}^\prime \right)
            + (-1)^s U_l^{m, -s} \sin \left( \theta_{lmp} + \theta_{lsq}^\prime \right) \right]
            & l -s\ even \\[0.5em]
            \frac{1}{2} \left[ U_l^{m, s} \cos \left( \theta_{lmp} - \theta_{lsq}^\prime \right)
            - (-1)^s U_l^{m, -s} \cos \left( \theta_{lmp} + \theta_{lsq}^\prime \right) \right]
            & l -s\ odd\ .
        \end{array} \right.
\end{align*}
Case $m$ odd, $l - m$ even:
\begin{align*}
    \mathcal{R}_k
    & = \sum\limits_{l \geq 2} \sum\limits_{m = 0}^l \sum\limits_{p = 0}^l (-1)^m\ F_{lmp} (I)
        \left( C_l^m \cos \theta_{lmp} + S_l^m \sin \theta_{lmp} \right) \\
    & = \sum\limits_{l \geq 2} \sum\limits_{m = 0}^l \sum\limits_{p = 0}^l \sum\limits_{s = 0}^l
        \sum\limits_{q = 0}^l (-1)^{m+s}\ (-1)^{k_1} A_l^{m, s} (l - s)! F_{lmp} (I) F_{lsq} (I_k) \\
    & \hspace{12pt} \times
        \left\{ \begin{array}{cl}
            -S_l^{m, s} \sin \theta_{lsq}^\prime \cos \theta_{lmp}
            + C_l^{m, s} \cos \theta_{lsq}^\prime \sin \theta_{lmp} & l -s\ even \\[0.5em]
            S_l^{m, s} \cos \theta_{lsq}^\prime \cos \theta_{lmp}
            + C_l^{m, s} \sin \theta_{lsq}^\prime \sin \theta_{lmp} & l -s\ odd
        \end{array} \right. \\
    & = \sum\limits_{l \geq 2} \sum\limits_{m = 0}^l \sum\limits_{p = 0}^l \sum\limits_{s = 0}^l
        \sum\limits_{q = 0}^l (-1)^{m+s}\ (-1)^{k_1} A_l^{m, s} (l - s)! F_{lmp} (I) F_{lsq} (I_k) \\
    & \hspace{12pt} \times
        \left\{ \begin{array}{cl}
            \frac{1}{2} \left[ U_l^{m, s} \sin \left( \theta_{lmp} - \theta_{lsq}^\prime \right)
            + (-1)^s U_l^{m, -s} \sin \left( \theta_{lmp} + \theta_{lsq}^\prime \right) \right]
            & l -s\ even \\[0.5em]
            \frac{1}{2} \left[ U_l^{m, s} \cos \left( \theta_{lmp} - \theta_{lsq}^\prime \right)
            - (-1)^s U_l^{m, -s} \cos \left( \theta_{lmp} + \theta_{lsq}^\prime \right) \right]
            & l -s\ odd
        \end{array} \right.
\end{align*}

The proof of \eqref{eq:Theta_odd} is given by
\begin{align*}
    \mathcal{R}_k
    & = \sum\limits_{l \geq 2} \sum\limits_{m = 0}^l \sum\limits_{p = 0}^l (-1)^m\ F_{lmp} (I)
        \left( -S_l^m \cos \theta_{lmp} + C_l^m \sin \theta_{lmp} \right) \\
    & = \sum\limits_{l \geq 2} \sum\limits_{m = 0}^l \sum\limits_{p = 0}^l \sum\limits_{s = 0}^l
        \sum\limits_{q = 0}^l (-1)^{m+s}\ (-1)^{k_1} A_l^{m, s} (l - s)! F_{lmp} (I) F_{lsq} (I_k) \\
    & \hspace{12pt} \times
        \left\{ \begin{array}{cl}
            -C_l^{m, s} \cos \theta_{lsq}^\prime \cos \theta_{lmp}
            - S_l^{m, s} \sin \theta_{lsq}^\prime \sin \theta_{lmp} & l -s\ even \\[0.5em]
            -C_l^{m, s} \sin \theta_{lsq}^\prime \cos \theta_{lmp}
            + S_l^{m, s} \cos \theta_{lsq}^\prime \sin \theta_{lmp} & l -s\ odd
        \end{array} \right. \\
    & = \sum\limits_{l \geq 2} \sum\limits_{m = 0}^l \sum\limits_{p = 0}^l \sum\limits_{s = 0}^l
        \sum\limits_{q = 0}^l (-1)^{m+s}\ (-1)^{k_1} A_l^{m, s} (l - s)! F_{lmp} (I) F_{lsq} (I_k) \\
    & \hspace{12pt} \times
        \left\{ \begin{array}{cl}
            \frac{1}{2} \left[ -U_l^{m, s} \cos \left( \theta_{lmp} - \theta_{lsq}^\prime \right)
            + (-1)^s U_l^{m, -s} \cos \left( \theta_{lmp} + \theta_{lsq}^\prime \right) \right]
            & l -s\ even \\[0.5em]
            \frac{1}{2} \left[ U_l^{m, s} \sin \left( \theta_{lmp} - \theta_{lsq}^\prime \right)
            - (-1)^s U_l^{m, -s} \sin \left( \theta_{lmp} + \theta_{lsq}^\prime \right) \right]
            & l -s\ odd\ .
        \end{array} \right.
\end{align*}
This concludes the proof.

\section{Proof of Lemma~\ref{lem:theta2}}\label{app:Lemma8}
We present the proof of Lemma~\ref{lem:theta2}.

The proof of \eqref{eq:Theta_unite} comes from a direct check of
the following cases.

\vskip.1in

If $m$ is even and $l - m$ is even, then $l$ is even and $l - 1$ is odd. Thus, $t = 1$ and:

$(i)$ if $l - s$ is even, then $s$ is also even, so that $k_2$ and $k_3$ are even, and $y_s = 0$;

$(ii)$ if $l - s$ is odd, then $s$ is also odd, so that $k_2$ and $k_3$ are odd, and $y_s = 1/2$.

Therefore, we arrive at the expression \eqref{eq:Theta_even} for $\Theta_{lmpsq}$, as required.

\vskip.1in

If $m$ is even and $l - m$ is odd or if $m$ is odd and $l - m$ is even, then it can easily be seen that $l - 1$ is even so that $t = 0$.
Consequently, $k_2 = 1$ and $k_3 = 0$, and:

$(i)$ if $l - s$ is even, then $s$ is odd giving $y_s = 1/2$;

$(ii)$ if $l - s$ is odd, then $s$ is even giving $y_s = 0$.

We then arrive at \eqref{eq:Theta_even_odd} for $\Theta_{lmpsq}$, as required.

\vskip.1in

If $m$ is odd and $l - m$ is odd, then $l$ is even and $l - 1$ is odd. Thus, $t = 1$ and:

$(i)$ if $l - s$ is even, then $s$ is also even, so that $k_2$ and $k_3$ are odd, and $y_s = 0$;

$(ii)$ if $l - s$ is odd, then $s$ is also odd, so that $k_2$ and $k_3$ are even, and $y_s = 1/2$.

Therefore, we arrive at the expression \eqref{eq:Theta_odd} for $\Theta_{lmpsq}$, as required.

This concludes the proof.

\section{Numerical validation of the lunisolar expansions}\label{app:orbit_comparison}

\begin{figure}[hpt]
\centering
\vglue0.1cm
\hglue0.1cm
\includegraphics[width=5truecm,height=4truecm]{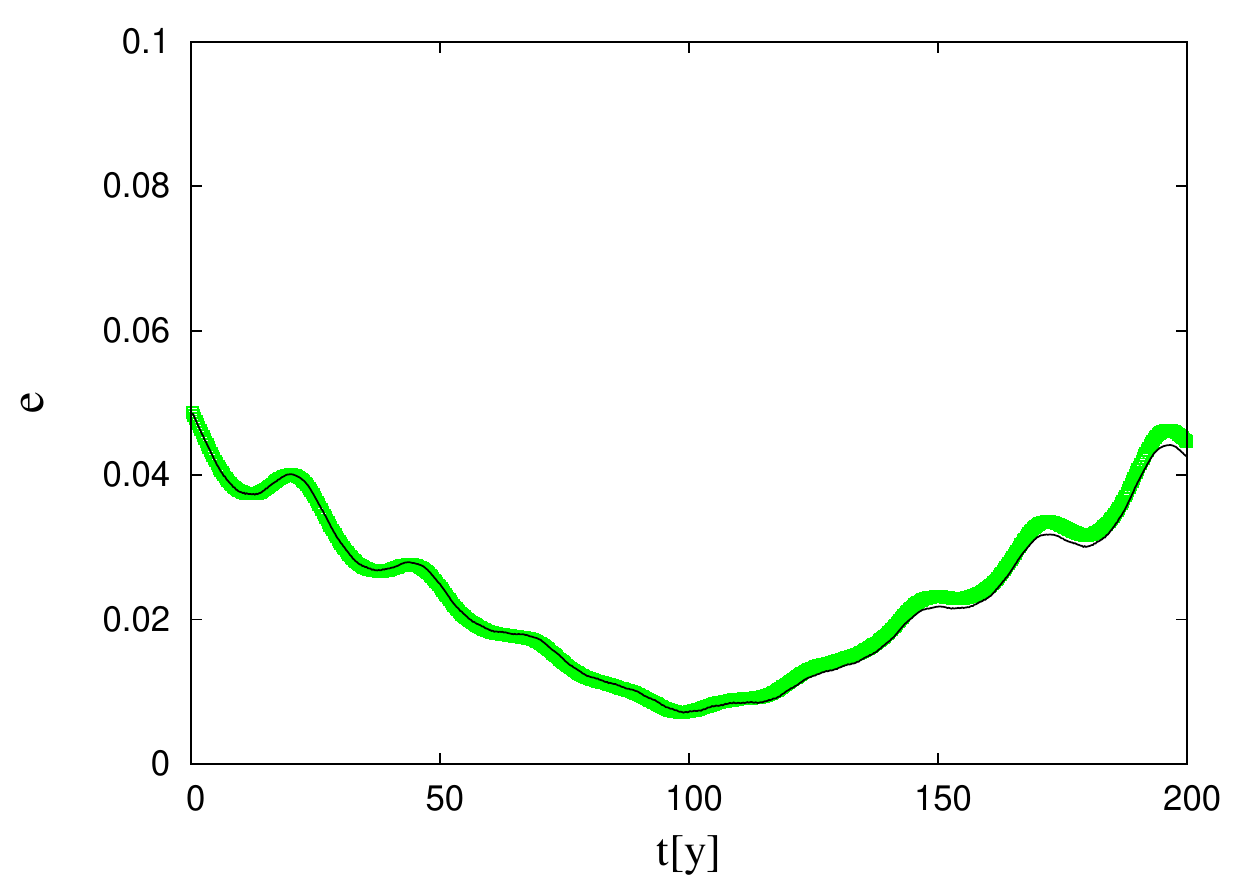}
\includegraphics[width=5truecm,height=4truecm]{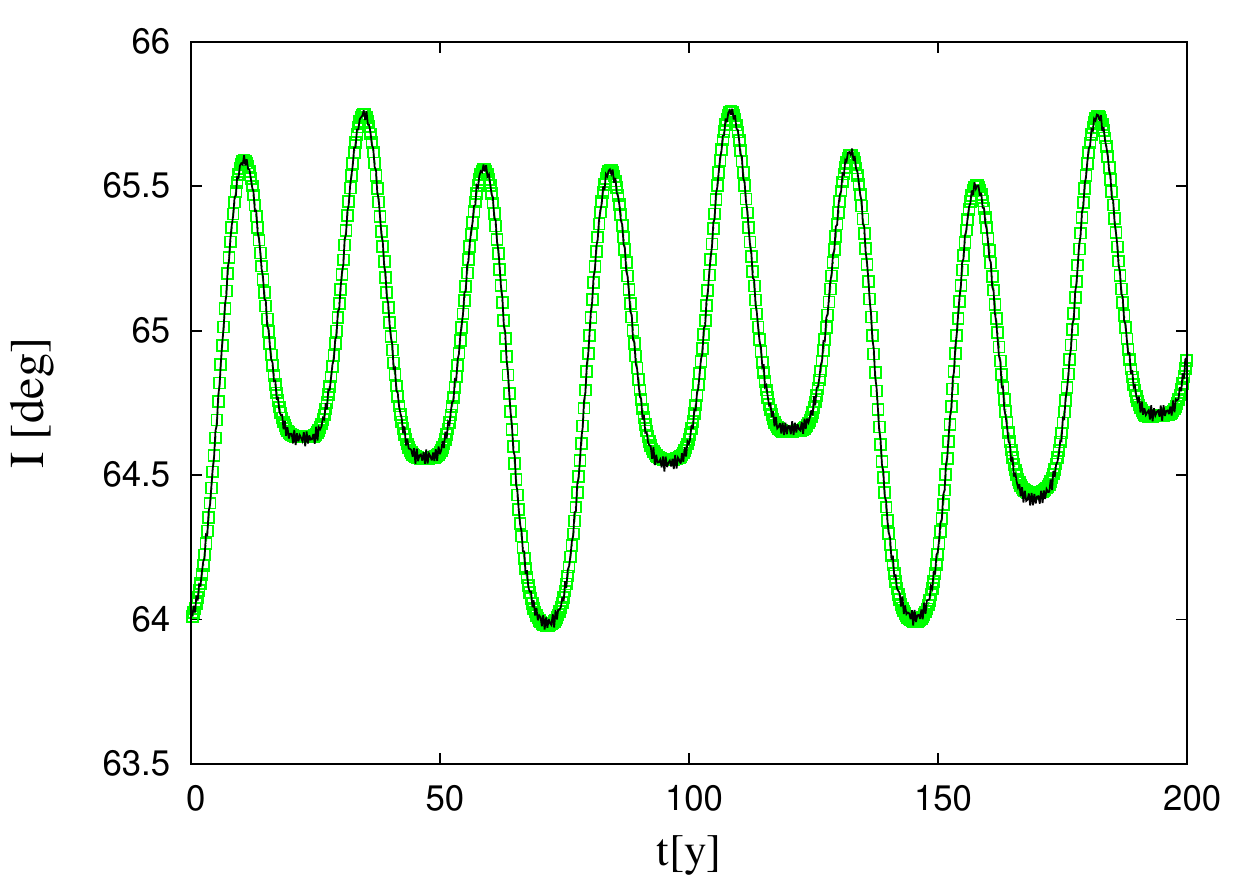}
\includegraphics[width=5truecm,height=4truecm]{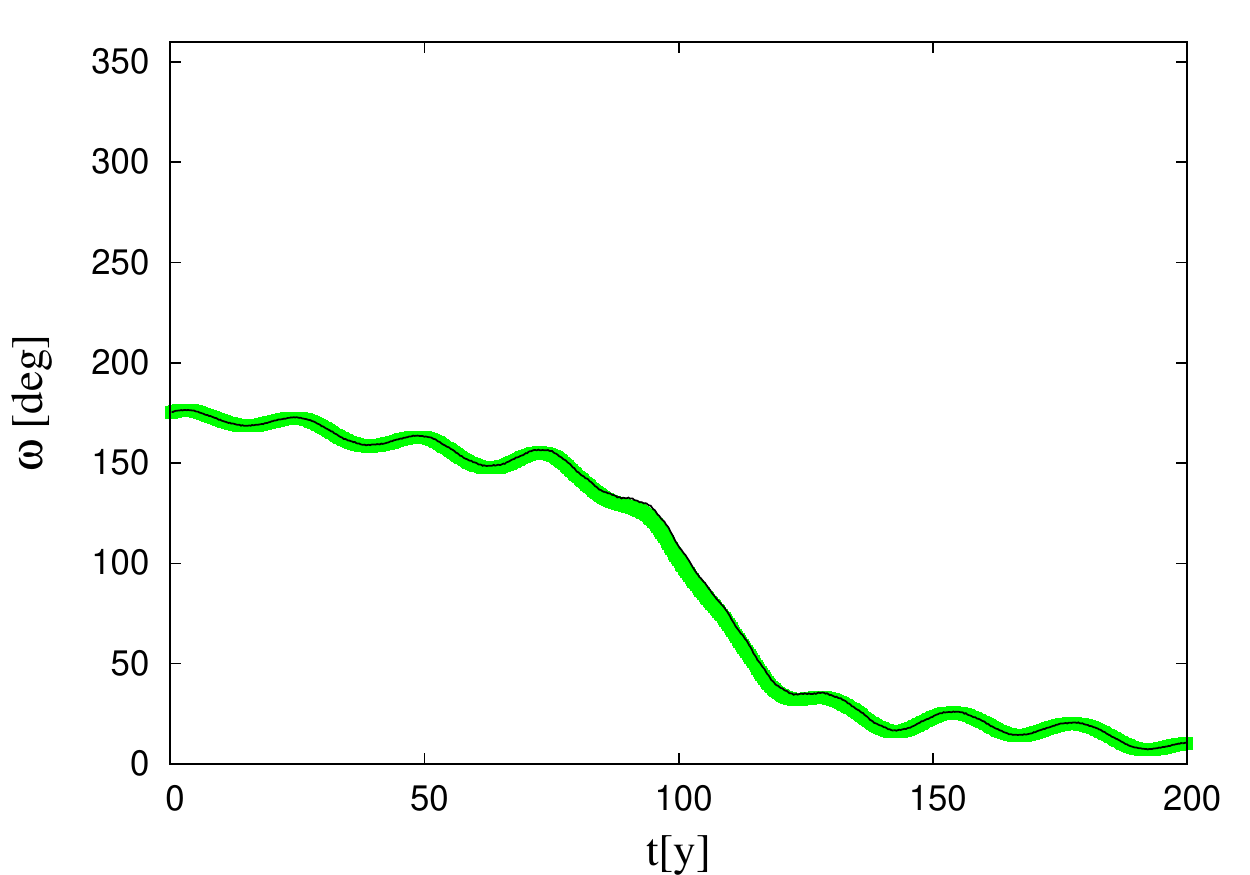}
\includegraphics[width=5truecm,height=4truecm]{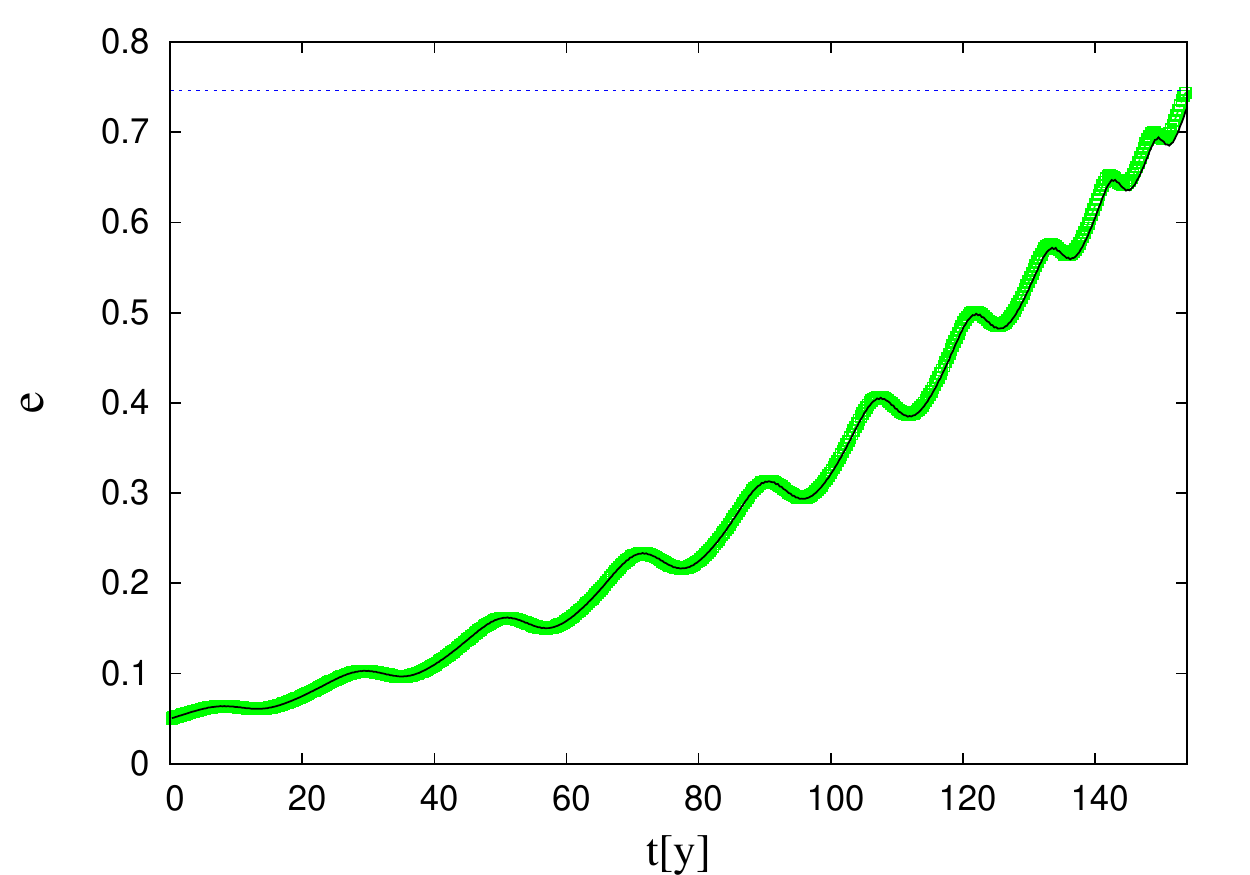}
\includegraphics[width=5truecm,height=4truecm]{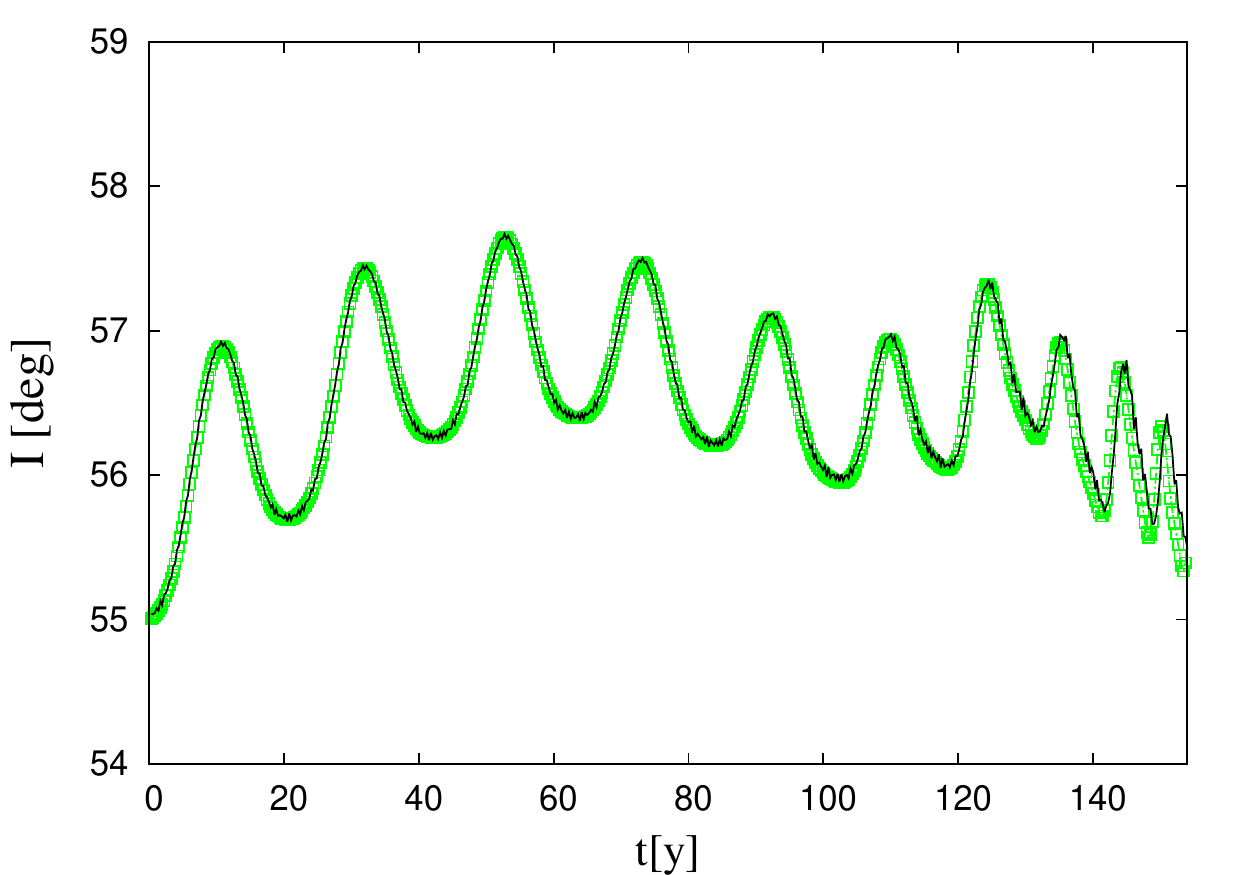}
\includegraphics[width=5truecm,height=4truecm]{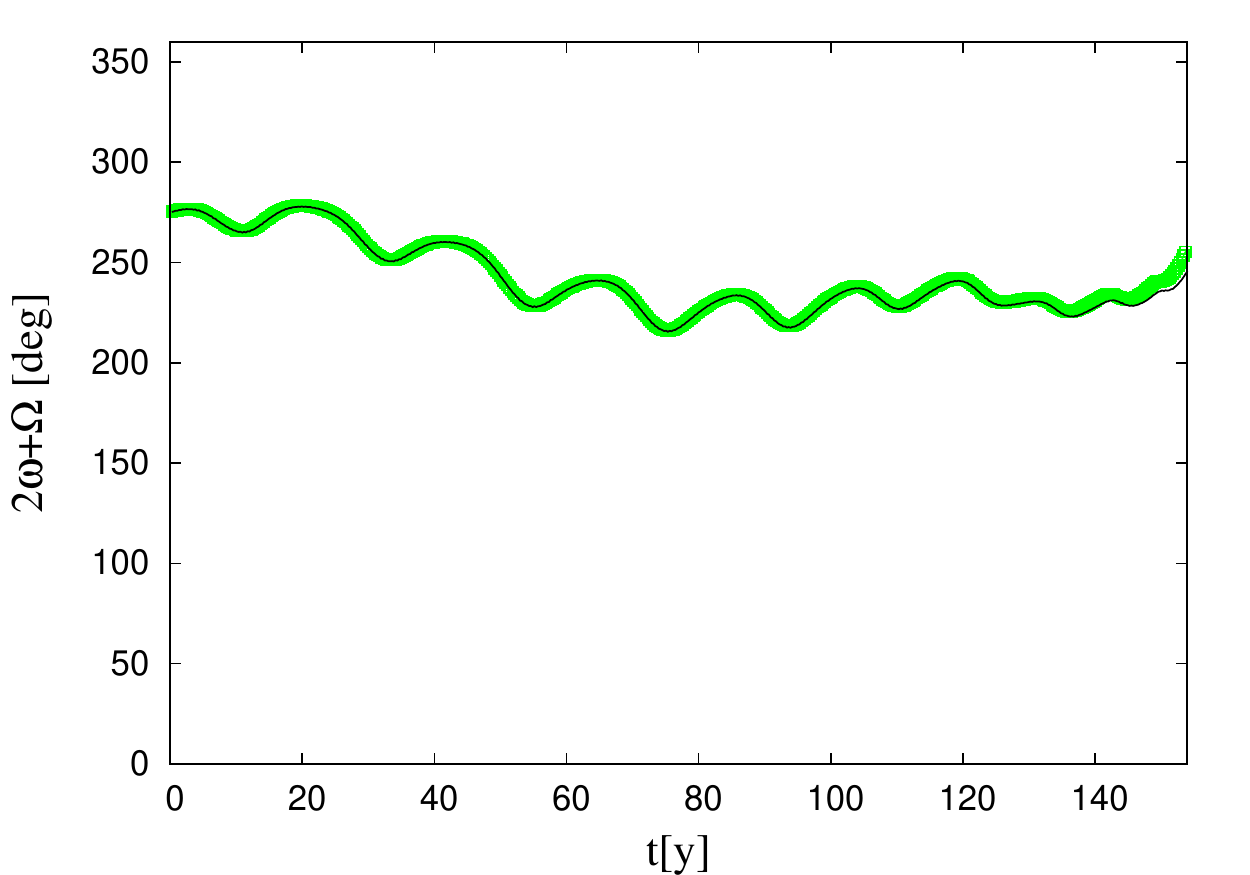}
\vglue0.6cm
\caption{Integration of an orbit within the libration region associated with
the critical inclination resonance (top panels) and an orbit inside
the resonance $2 \dot{\omega}+\dot{\Omega}=0$ (bottom plots).
We provide the plots for the eccentricity (left column), inclination
(middle column) and resonant angle (right column) as a function of time (in years).
The initial data are the following: $a=24293$ km, $e(0)=0.049$, $I(0)=64^{\circ}$,
$\omega(0)=175^\circ$, $\Omega(0)=150^\circ$ for the top panels, and $a=25271$ km, $e(0)=0.05$,
$I=55^\circ$, $\omega(0)=50^\circ$, $\Omega(0)=175^\circ$ for the bottom plots.
The initial epoch for both orbits is J2000. The green color (thicker line)
is used for the analytical model described in Section~\ref{sec:averages};
it includes the Earth's gravity harmonic $J_2$, the disturbing functions due to Sun and Moon, averaged
over both anomalies (of the satellite and the third body perturber) with
the expansion of the Moon taken up to degree $l=3$. The black color (thinner line)
corresponds to a Cartesian model, which includes the Earth's gravity harmonics up to degree
and order $2$, as well as the attraction of Sun and Moon (see \cite{CGmajor,CGminor}).
The horizontal line in the left bottom plot indicates the eccentricity value leading to re-entry.
}
\label{fig:two_orbits}
\end{figure}

To validate the lunisolar expansions, we compare in Figure~\ref{fig:two_orbits} the results obtained by using
a Newtonian (Cartesian) model and an analytical model. The Newtonian model includes the Earth's gravity harmonics
up to degree and order $2$, as well as the attraction of the Sun and Moon (see \cite{CGmajor,CGminor} for further details);
the analytical model is presented in Section~\ref{sec:averages}, based on the expansions \eqref{Rsun} and \equ{eq:lane}.
Osculating elements have been used for integrating the Newtonian model, while mean elements
have been used for the analytical model. Although we made several tests with
different dynamical conditions and different initial data, we present the
results for two orbits: one located inside a libration region corresponding to the critical inclination resonance
and the other placed inside a resonant island associated with the resonance $2 \dot{\omega}+\dot{\Omega}=0$.
The results show that for small and large eccentricities, as well as for different resonances,
the two approaches lead to similar results, thus yielding a further validation of the lunisolar expansions
presented in Section~\ref{sec:Kaula_development} and \ref{sec:Lane_development}.

\end{appendices}

\vskip.2in

\section*{Acknowledgements}
We are grateful to an anonymous referee for very helpful comments, which allowed us to improve
the content and the presentation.

A.C., G.P., A.R. are partially funded by the European Commission's Framework Programme 7, through
the Stardust Marie Curie Initial Training Network, FP7-PEOPLE-2012-ITN, Grant Agreement 317185.
A.C. was partially supported by PRIN-MIUR 2010JJ4KPA$\_$009 and GNFM/INdAM.
C.G. was supported by a grant of the Romanian National Authority for
Scientific Research and Innovation, CNCS - UEFISCDI, project number
PN-II-RU-TE-2014-4-0320 and by GNFM/INdAM. G.P. was partially supported by GNFM/INdAM.

\bibliographystyle{spmpsci}

\end{document}